\documentclass[12pt]{amsart}
\usepackage{amssymb,amsmath,amsthm,amsfonts}
\usepackage{thmtools, thm-restate}
\usepackage{enumitem}
\usepackage{xspace}
\usepackage{graphicx}
\usepackage[dvipsnames]{xcolor}
\usepackage{comment}
\usepackage{caption,float}
\usepackage{hyperref}
\usepackage{fullpage}

\newcommand{\eqdef}{\ensuremath{:\!=}}
\renewcommand{\vec}[1]{\mathbf{#1}}

\newcommand{\dunion}{\sqcup}

\declaretheorem[name=Theorem, numberwithin=section]{thm}
\declaretheorem[name=Definition,style=definition,sibling=thm]{defi}

\declaretheorem[name=Lemma, sibling=thm]{lem}
\declaretheorem[name=Claim, sibling=thm]{claim}
\declaretheorem[name=Corollary, sibling=thm]{crl}
\declaretheorem[name=Remark,style=remark,numbered=no]{rk}
\declaretheorem[name=Example,style=remark,sibling=thm]{eg}
\declaretheorem[name=Conjecture,style=definition,numbered=no]{conjecture}

\newcommand{\stt}[2]{\ensuremath{\left\{#1\,:\,#2\right\}}}
\newcommand{\crd}[1]{\ensuremath{\left|#1\right|}}

\newcommand{\F}{\mathbb{F}}
\DeclareMathOperator{\codim}{codim}
\DeclareMathOperator{\chr}{char}
\DeclareMathOperator{\spn}{span}

\newenvironment{mx}{\left(\begin{matrix}}{\end{matrix}\right)}
\DeclareMathOperator{\per}{per}
\DeclareMathOperator{\poly}{poly}
\DeclareMathOperator{\Per}{Per}

\DeclareMathOperator{\rank}{rank}
\DeclareMathOperator{\pk}{perrank}
\newcommand{\prk}[1]{\pk(#1)}

\newcommand{\tr}[1]{{#1}^{T}}

\newcommand{\frnd}{jointly-permanull\xspace} 
\newcommand{\vgood}{permanull\xspace} 
\newcommand{\SCAZ}{\mathsf{SC0}}
\newcommand{\WS}{\mathsf{WS}}
\newcommand{\ellstar}{{\ell^{*}}}

\begin{document}
\title[Permanental ranks over a finite field]{Permanental rank versus determinantal rank of random matrices over finite fields}

\author{Fatemeh Ghasemi}
\address{Fatemeh Ghasemi, Department of Mathematics, University of Toronto}
\email{fatemeh.ghasemi@mail.utoronto.ca}

\author{Gal Gross}
\address{Gal Gross, Department of Mathematics, University of Toronto}
\email{g.gross@mail.utoronto.ca}

\author{Swastik Kopparty}
\address{Swastik Kopparty, Department of Mathematics and Department of Computer Science, University of Toronto}
\email{swastik.kopparty@utoronto.ca}
\thanks{Research supported by an NSERC Discovery grant.}

\date{\today}

\begin{abstract}
This paper is motivated by basic complexity and probability questions about permanents of random matrices over small finite fields, and in particular, about properties separating the permanent and the determinant. 

Fix $q = p^m$ some power of an odd prime, and let $k \leq n$ both be growing.
For a uniformly random $n \times k$ matrix $A$ over $\F_q$, we study the probability that all $k \times k$ submatrices of $A$ have zero permanent; namely that $A$ does not have full {\em permanental rank}.

When $k = n$, this is simply the probability that a random square matrix over $\F_q$ has zero permanent, which we do not understand. We believe that the probability in this case is $\frac{1}{q} + o(1)$, which would be in contrast to the case of the determinant, where the answer is $\frac{1}{q} + \Omega_q(1)$.

Our main result is that when $k$ is $O(\sqrt{n})$, the probability that a random $n \times k$ matrix does not have full permanental rank is essentially the same as the probability that the matrix has a $0$ column, namely $(1 +o(1)) \frac{k}{q^n}$. In contrast, for determinantal (standard) rank the analogous probability is $\Theta(\frac{q^k}{q^n})$.

At the core of our result are some basic linear algebraic properties of the permanent that distinguish it from the determinant.

\end{abstract}

\subjclass[2020]{Primary 15A15; Secondary 15B33, 15B52.}

\keywords{Permanent, random matrices over a finite field}

\maketitle

\section{Introduction}\label{sec:motivation}

Like the determinant, the permanent of an $n\times n$ matrix $A = [a_{ij}] \in M_{n}(\F)$ is defined as a sum over permutations:
$$\per A := \sum_{\sigma \in S_n}\prod_{i=1}^{n}a_{i\sigma(i)}.$$

Despite their apparent similarity, the permanent function and the determinant function conjecturally lie at vastly different locations in the computational landscape. It is thus of great interest to prove unconditional separations between their properties.

Unlike the determinant, which can be computed efficiently (e.g., via Gaussian elimination), the permanent is notoriously difficult to compute.\footnote{Unless we are in characteristic $2$, in which case the permanent equals the determinant.} This has been vaguely understood for a long time, starting with P\'olya's question \cite{Pol13} about the impossibility of reducing the computation of a permanent to that of a determinant. More recently, Valiant \cite{Val79} defined the counting complexity class $\#\mathsf{P}$, proved that this complexity class is at least as hard as $\mathsf{NP}$, and that the permanent of a matrix with entries in $\{0, 1\}$ is $\#\mathsf{P}$-complete. He also showed that over fields of odd characteristic $p$, the permanent is complete for the modular counting class $\mathsf{MOD}_p\mathsf{P}$.

In a separate line of work, Valiant \cite{Val79b} defined algebraic analogues of the $\mathsf{P}$ and $\mathsf{NP}$ complexity classes---called $\mathsf{VP}$ and $\mathsf{VNP}$---and proved that computing the permanent is $\mathsf{VNP}$-complete. In contrast, the determinant lies in $\mathsf{VP}$ (and is complete for the further subclass $\mathsf{VBP}$).
This makes the permanent one of the main candidates for an explicit polynomial which may separate $\mathsf{VP}$ from $\mathsf{VNP}$ and so it has been the subject of intense research in algebraic complexity. In particular, this has been studied under the title of ``permanent vs.~determinant", seeking to understand the intrinsic differences between the permanent polynomial and the determinant polynomial.

The permanent has also played an important role in average case complexity. Over fields of size $\poly(n)$, building on the work of Beaver-Feigenbaum~\cite{BF90}, Lipton~\cite{LIP91} showed the random self reducibility of permanents under the uniform distribution. This implies that if there is no algorithm that efficiently computes the permanent of arbitrary $n \times n$ matrices (this is widely believed, since the permanent over $\F_p$ for $p$ prime is a $\mathsf{MOD}_p\mathsf{P}$-hard problem), then there is no algorithm that can efficiently compute the permanent even if it is allowed to err on some $\frac{1}{\poly(n)}$ fraction of $n \times n$ matrices. Using error-correcting properties of low-degree polynomials over large fields, the size of this fraction was improved by Gemmel-Lipton-Rubinfeld-Sudan-Wigderson~\cite{Gem+91} (to $\Omega(1)$), Gemmel-Sudan~\cite{GS92} (to $1/2 - o(1)$), and Cai-Pavan-Sivakumar~\cite{CPS99} (to $1-o(1)$)\footnote{Average case hardness results under stronger assumptions were given by Feige-Lund~\cite{FL96}. The list version of the question was proved hard by Cai-Hemachandra~\cite{CH91}.}. While these results do not hold over smaller fields, Feigenbaum-Fortnow~\cite{FF93} showed that for fields $\F_q$ with $q$ a fixed prime power, there is a nonuniform distribution under which this kind of average case hardness holds. Today it is not known whether for a fixed prime $p$, computing the permanent for $99\%$ of all $n \times n$ matrices over $\F_p$ is $\mathsf{MOD}_p\mathsf{P}$-hard. This gives further motivation for studying the distribution of the permanent of a uniformly random matrix over $\F_p$ which, surprisingly, is not well understood.

If the entries of an $n\times n$ matrix $A$ are chosen uniformly at random from a finite field $\F_q$ of characteristic $p$ and order $q = p^m$, it is well known that the probability of $\det A$ being $0$ is
\[1-\prod_{k=0}^{n-1}\left(1-\frac{q^k}{q^n}\right) = \frac{1}{q}+\Omega(q^{-2}).\]
A simple argument shows that the exact same expression is also an upper bound for the probability that $\per A$ is $0$, and other than a $(1 +o(1))$ factor improvement~\cite{BG12}, this is the best known upper bound.  We believe, however, that something much stronger is true; that the permanent behaves like a random polynomial.

\begin{conjecture}
Let $p$ be an odd prime and $q = p^m$ some fixed power. Let $A$ be a uniformly random $n \times n$ matrix with entries in $\F_q$. Then:
\[\Pr[\per A = 0 ] = \frac{1}{q} + o(1).\]
\end{conjecture}

Since the permanent function enjoys a cofactor expansion formula similar to that of the determinant, one simple way for the permanent to vanish is if each of its cofactors vanishes. This observation naturally leads to the consideration of the \emph{permanental rank}. First introduced by Yu \cite{Yu99} in connection with the Alon-Jaeger-Tarsi conjecture \cite{AT89}, the \emph{permanental rank} $\prk{A}$ of an $n\times k$ matrix $A$ is the largest $r$ such that $A$ has an $r\times r$ submatrix with nonzero permanent. Note that if `nonzero permanent' is replaced with `nonzero determinant', the result is an equivalent definition of the usual notion of the rank of a matrix, $\rank A$. For $k = n$ we have $\per A = 0 \iff \prk{A} < k$, i.e. if $A$ is of non-full permanental rank. 

We have the simple observation, based on cofactor expansion, that if $A$ is a uniformly random $n \times n$ matrix, and $B$ is a uniformly random $n \times (n-1)$ matrix, then:
$$\Pr[\per A = 0] = \frac{1}{q} + \left( 1- \frac{1}{q}\right)\cdot \Pr[\prk{B} < n-1].$$
This relates the conjecture to the problem of understanding the probability that the rectangular matrix $B$ has non-full permanental rank: the above conjecture is equivalent to the statement that $\Pr[\prk{B} < n-1] = o(1)$.

\subsection{Main Result}
In this paper we prove bounds on the probability that a random $n\times k$ matrix with $k \ll n$ has non-full permanental rank, i.e., the probability that all $k\times k$ submatrices have zero permanents. 

The trivial way such an $n\times k$ matrix $A$ may fail to have full permanental rank is if it has a $0$-column. Viewing the columns as random vectors in $\F_q^n$, the probability of a zero column is easily seen to be $kq^{-n} + O((kq^{-n})^2)$. Therefore, the probability $Z$ of $\prk{A}<k$ is at least as large:
$$\Pr[Z] \geq \frac{k}{q^{n}}+o(q^{-n}).$$
Our main result, Theorem \ref{thm:main}, proves a matching upper bound for $k = O(\sqrt{n})$:
$$\Pr[Z] = \frac{k}{q^{n}}+o(q^{-n}).$$
Thus for $k = O(\sqrt{n})$ the only significant obstruction to full permanental rank is the trivial one of a zero column.

For determinants, we note the probability that a random $n \times k $ matrix has non-full determinantal (i.e. standard) rank is about $\Theta_q(\frac{q^k}{q^n})$ --- it is within a constant factor of the probability that the last column lies in the span of the first $k-1$ columns. This establishes a separation in the behaviours of the permanental rank and the (determinantal) rank. We hope that our ideas can lead to a proof of our conjecture above.

\subsection{Techniques}
Our probabilistic argument relies on novel structural results. We introduce the notion of a \emph{\vgood subspace}, which is a subspace $S \leq \F_q^n$ such that any $n\times n$ matrix $A$ whose columns lie in $S$ has vanishing permanent, $\per A = 0$. Theorem \ref{thm:c1vgood} shows that the only codimension-1 \vgood subspaces are the trivial ones $\{\vec{e}_i\}^\perp$. Generalizing this notion, we call a list of $n$-subspaces $S_1, \ldots, S_n$ \emph{\frnd} if every $n\times n$ matrix $A$ such that the $i$-th column of $A$ is an element of $S_i$ has $\per A = 0$.\footnote{Permuting the columns of $A$ does not affect the vanishing of the permanent, so one could consider $n$-\emph{multisets} of subspaces instead. We opted for lists for ease of readability.} Our main structural result, Theorem \ref{thm:manyfriends}, generalizes Theorem \ref{thm:c1vgood} by showing that if $n \geq 3$ and each $S_i$ has codimension at most $1$ then $S_1 = \cdots = S_n = \{\vec{e}_i\}^\perp$.

Both these facts do not hold for the determinant. This is the key source of how our results distinguish between the permanent and the determinant. Understanding when spaces $S_1, \ldots, S_n$ of codimension $2$ or larger are \frnd is very interesting, and seems highly relevant for relaxing our requirement that $k \leq O(\sqrt{n})$. In Sections \ref{sec:0fld} and \ref{sec:poly} we make some small steps in this direction.

\subsection{Other related work}
For characteristics $p\geq 3$, Dolinar et al.~\cite{DGKO} have shown that for a sufficiently large order $q$, the probability $\Pr[\per A = 0]$ is strictly less than $\Pr[\det A = 0]$ by recursively computing upper bounds for the number of matrices with vanishing permanents; Budrevich and Guterman \cite{BG12} removed the condition that the order of the field be sufficiently large. Several papers improved the upper bounds on the number of matrices with vanishing permanents, the most recent being Budrevich's paper \cite{Bud18} (see also the references therein). 
These results present significant progress in understanding the vanishing behaviour of the permanent, but are not yet sufficient to prove that for fixed $q$ and growing $n$, the difference between the probability that the permanent of a random matrix is zero and the probability that the determinant of a random matrix is zero is  at least a positive constant.

Very recently, Scheinerman~\cite{Scheinerman} developed new methods for computing astonishingly large permanents over the field $\F_3$ on present day computers. He used his algorithms to compute permanents of random $30 \times 30$ matrices over $\F_3$, and found that the probability of zero permanent does seem to be very close to $1/3$.

Over the real numbers, Tao and Vu~\cite{TV09} showed that the permanent of a random $\pm1 $ matrix is almost surely nonzero (and they further determined its magnitude). More recently, Kwan and Sauermann~\cite{KS22} extended these results to symmetric random $\pm 1$ matrices. There seem to be some quantitative forms of permanental rank that appear in these works; it would be interesting to see if there are tools from there that can help resolve our main conjecture.

\section*{Acknowledgements}

SK would like to thank Noga Alon, Jeff Kahn, Danny Scheinerman, Srikanth Srinivasan, and Sergey Yekhanin  for valuable discussions about these topics over the past many years. We thank the anonymous referees for pointers to the literature and for suggestions that improved the presentation.

\section{Results}
In this section we state our main theorem, showing that for $k = O(\sqrt{n})$, a random $n\times k$ matrix has non-full permanental rank  with probability about $k/q^n$. 

\begin{restatable}{thm}{prbthm}
\label{thm:main}
 Let $k, n \in \mathbb N$ with $k \leq 0.1\sqrt{ n}$, .
 
 Let $X$ be a uniformly random $n\times k$ matrix with entries in $\F_q$. Let $Z$ be the event that all $k \times k$ submatrices of $X$ have permanent equal to $0$. Then:
 $$ \Pr[Z] =   \frac{k}{q^n} + O(q^{-1.1n}).$$
\end{restatable}

\begin{rk}
    For much smaller $k$, the error term gets smaller. For example, for $k = o(\sqrt{n})$, we have
    $$\Pr[Z] \leq \frac{k}{q^n} + O(q^{-(2-o(1))n}).$$ The precise relationship is given in the proof of Claim \ref{claim:bigprb} below.
\end{rk}

The main ingredient in our proof is a structural result about the permanent. In order to state the result we need the following definitions.

\begin{defi}
     The list of subspaces $S_1, \ldots, S_n \leq \F^n$ is said to be \emph{\frnd} if every $n\times n$ matrix $A$ whose $i$-th column is a vector in $S_i$ (for $1 \leq i \leq n$) has $\per A = 0$.
\end{defi}

Our main structural result classifies codimension-$1$ \frnd lists.

\begin{restatable*}[Classification of codimension-$1$ \frnd lists]{thm}{manyfriends}
\label{thm:manyfriends}
    Let $n\geq 3$ and $S_1, \ldots, S_n \leq \F_q^n$ be subspaces, each of dimension $\geq n-1$. Then $S_1, \ldots, S_n$ is a \frnd list if and only if $S_1 = S_2 = \cdots = S_n = \{\vec{e}_i\}^\perp$ for some $1 \leq i \leq n$.
\end{restatable*}

Note that the ``if'' direction of the theorem is obvious, the interesting part is the ``only if'' direction. A particularly interesting aspect of this theorem is that it gives a separation between the permanent and the determinant: $n$ subspaces in $\F^n$ of codimension at most $1$ are jointly ``determinull'' if an only if they are the same codimension $1$ subspace.

The proof of Theorem \ref{thm:manyfriends} proceeds by induction; the base case $n=3$ is Lemma \ref{lem:3friends}, and relies on some direct computations. For $n=2$ the statement is false!
The inductive step, remarkably, only uses the cofactor expansion formula of the permanent, and would have also gone through for the determinant --- it is the base case that makes all the difference! The inductive step relies on a simple graph-theoretic observation which is Lemma \ref{lem:nostar}. In the next section we show how Theorem \ref{thm:main} follows from \ref{thm:manyfriends}. We defer the proof of Theorem \ref{thm:manyfriends} to Section \ref{sec:structural}.

Understanding \frnd subspaces of larger codimension looks like a
very interesting problem, and seems relevant to extending the range of $k$ in our main theorem. An interesting special case of \frnd subspaces is when all the subspaces in the list are the same subspace $S$.

\begin{defi}
    The subspace $S \leq \F^n$ is said to be \emph{\vgood} if every $n\times n$ matrix $A$ whose columns are elements of $S$ has $\per A = 0$.
\end{defi}

This case plays an important role in the theoretical development leading to the proof of our main structural result. In Sections~\ref{sec:0fld} and~\ref{sec:poly},  we characterize permanullness in larger codimension. In Section \ref{sec:0fld}, we show that over fields of sufficiently high characteristic, the only \vgood subspaces are the trivial ones, i.e.~the ones having a common zero coordinates (Theorem \ref{thm:trivialpermanull}). In Section \ref{sec:poly} we give a full characterization of \vgood subspaces (for any characteristic of the ambient field) in terms of an auxiliary polynomial which we call the \emph{permanental polynomial} (Theorem \ref{thm:zeropoly}). This allows us to check if a given codimension-$k$ subspace is \vgood with $O_k(n^k)$ field operations, and a general subspace with $2^{O(n)}$ field operations.
\section{Proof of the Main Theorem}\label{sec:main}
In the current section we show how to use Theorem \ref{thm:manyfriends} to prove Theorem \ref{thm:main}.

\begin{proof}[Proof of Theorem \ref{thm:main}]
    Let $\SCAZ$ be the event that some column of $X$ is the zero vector $\vec{0}$. We will estimate the probability of $Z$ by showing that it almost coincides with the event $\SCAZ$ (up to probability $O(q^{-1.1n})$).
 
    Note that $\SCAZ \implies Z$. To get some sort of reverse implication, we will define a ``well-spreadedness'' event $\WS$ and show that:
    \begin{itemize}
        \item  $\WS$ occurs with probability $1 - o(q^{-n})$. 
        \item $\WS \wedge Z \implies \SCAZ$.
    \end{itemize}
 
    To deduce $\SCAZ$ in the second step, we will use the classification of codimension-$1$ \frnd lists, Theorem \ref{thm:manyfriends}. The event $\WS$ will help us meet the hypothesis of Theorem \ref{thm:manyfriends}. We now work towards defining $\WS$.
 
    Let $\vec{x}_1, \ldots, \vec{x}_n \in \F_q^k$ be the rows of $X$. We first observe that if $Z$ happens, any partition of the rows of $X$ into $k$ parts will  give us a construction of a \frnd list. Indeed, if $N_1, \ldots, N_k$ is a partition of $[n] = \{1, 2, \ldots, n\}$, then the linear spaces $S_1, \ldots, S_k \leq \F_q^k$ defined by
    	\[S_i = \spn\{ \vec{x}_j : j \in N_i \}\]
    form a \frnd list; this follows immediately from the hypothesis that any $k\times k$-submatrix of $X$ has vanishing permanent, by the facts that the permanent is a multilinear function of the rows. To apply Theorem \ref{thm:manyfriends} to this, we will need the subspaces $S_i$ to have dimension at least $k-1$. This motivates the following definition of the well-spreadedness event $\WS$: it is the event that there exists a partition $N_1, \ldots, N_k$ of $[n]$ such that for each $i$ we have $\dim(S_i) \geq k-1$. It is immediate from Theorem \ref{thm:manyfriends} that $\WS \wedge Z \implies \SCAZ$.

    We make the following claim, which we prove immediately after showing how to use it to complete current proof.

    \begin{claim}\label{claim:bigprb}
        $$\Pr[\WS] \geq 1 - O(q^{-1.1n}).$$
    \end{claim}
 
    Assuming the claim, we complete the proof of the theorem. By the claim, we have:
        $$\Pr[Z] \leq \Pr[\WS \wedge Z ] + \Pr[\neg \WS ] \leq \Pr[\SCAZ] + O(q^{-1.1n}).$$
    On the other hand, since $\SCAZ \implies Z$, we have:
        $$\Pr[\SCAZ] \leq \Pr[Z].$$
 
    Finally, since $\Pr[\SCAZ] = \frac{k}{q^n} + O(k^2 \cdot q^{-2n})$, we get the desired result.
\end{proof}

\begin{proof}[Proof of Claim \ref{claim:bigprb}]
    We give an algorithm for finding the partition $N_1, \ldots, N_k$ witnessing the well-spreadedness event $\WS$, and then show that when $k$ is small enough compared to $n$, the algorithm succeeds with very high probability.
 
    \begin{enumerate}
        \item Initialize $N_1 = N_2 = \cdots = N_k = \emptyset$.
        \item Initially all $i \in [k]$ are marked \textsc{active}.
        \item For each $i \in [k]$, let $S_i$ denote $\spn\{ \vec{x}_j : j \in N_i  \}$.
        \item For $\ell = 1, 2, \ldots, n$, do the following:
        \begin{itemize}
            \item Pick one \textsc{active} $i \in [k]$  for which $\vec{x}_\ell$ does not lie in $S_i$, and:
            \begin{itemize}
                \item Add $\ell$ into the set $N_i$.
                \item If $\dim(S_i) \geq k-1$, then mark $i$ \textsc{inactive}.
            \end{itemize}
            \item If no such \textsc{active} $i$ exists, mark $\ell$ as \textsc{ineffective}, and we add $\ell$ into one of the $N_i$.
        \end{itemize}
    \end{enumerate}

    We now prove that with very high probability by the end of the algorithm each $S_i$ has dimension $\geq k-1$. This will prove that with very high probability $\WS$ occurred.
 
    Observe that in any iteration of the $\ell$ loop, either $\sum_{i \in [k]} \dim(S_i)$ increases by $1$, or $\ell$ is \textsc{ineffective}. Thus at the end of the algorithm, if some $i \in [k]$ has $\dim(S_i) \leq k-2$ (so that $i$ is \textsc{active}), then at most $(k-1)^2+k-2 = k(k-1)-1$ of the $\ell$'s were ``effective'', and we must have strictly more than $n-k(k-1)$ \textsc{ineffective} $\ell$'s.

    Fix some $\ellstar \in [n]$. Suppose we are just about to start the $\ellstar$-th iteration of the loop $\ell \leftarrow \ellstar$. So far only $\vec{x}_1, \ldots, \vec{x}_{\ellstar -1}$ have been revealed. What would it take for $\ellstar$ to be marked \textsc{ineffective}?

    This will happen only if $\vec{x}_\ellstar$ lies inside $S_i$ for every \textsc{active} $i$, namely if 
        $$\vec{x}_{\ellstar} \in \bigcap_{i\mbox{\,\scriptsize \textsc{active}}} S_i.$$
 
    Observe that for any \textsc{active} $i$, we have $\dim(S_i) \leq k-2$. Thus, if there is at least one \textsc{active} $i$, the probability that $\ellstar$ is marked \textsc{ineffective} is at most:
        \[\frac{q^{k-2}}{q^k} = q^{-2}.\]
    If there are no \textsc{active} $i$, then we already have $\dim(S_i) \geq k-1$ for all $i$.
 
    Suppose we have $L>n-k(k-1)$ \textsc{ineffective} $\ell$'s. Using the union bound, the probability that some $i$ remains \textsc{active} at the end of the algorithm is at most
    \begin{align*}
        \binom{n}{L} \cdot \left(q^{-2} \right)^{L} &= \binom{n}{n-L} \cdot q^{-2L}\\
        &\leq \binom{n}{n-L} \cdot q^{-2(n-k(k-1))}\\
        &\leq \binom{n}{k(k-1)}q^{2k(k-1)}q^{-2n}\\
        &\leq  \binom{n}{k^2}q^{2k^2}q^{-2n}
    \end{align*}
    Thus for $k = 0.1 \sqrt{n}$, the probability that some $i$ remains \textsc{active} is at most $O(q^{-1.1n})$.
    
    If none of the $i$ is \textsc{active} at the end of the algorithm, then the partition $N_1, \ldots, N_k$ produced by the algorithm is a witness for the event $\WS$. Thus 
    $$\Pr[\WS] \geq 1 - \Pr[\mbox{Some $i$ remains \textsc{active}} ] \geq 1 - O(q^{-1.1n}).$$
\end{proof}

\section{Structural Results}\label{sec:structural}
In this section we prove the structural results we have relied on in the proof of Theorem \ref{thm:main}. Our main structural result is 

\manyfriends*

Note that the ``if'' direction of the theorem is obvious, and below we shall prove the ``only if'' direction. The proof proceeds by induction; the base case $n=3$ is Lemma \ref{lem:3friends}. The inductive step relies on a simple graph-theoretic observation which is Lemma \ref{lem:nostar}. In the remainder of this section we prove the necessary lemmas and show how Theorem \ref{thm:manyfriends} follows from them.

\subsection{Codimension-1 \vgood subspaces}

In this subsection, we focus on the case where all the subspaces $S_i$ are the same, namely subspaces $S \subseteq \F^n$ which are {\em \vgood}.

Clearly $S = \F^n$ cannot be \vgood (e.g., consider the identity matrix), which leads naturally to the question of whether there are any \vgood codimension-$1$ spaces. Since a matrix with a $\vec{0}$-row has a vanishing permanent, we see that any subspace of the form $\{\vec{e}_i\}^\perp$ is \vgood. We now show that these are the only codimension-$1$ \vgood spaces.

\begin{restatable}[Classification of codimension-$1$ \vgood subspaces]{thm}{vgoodone}
\label{thm:c1vgood}
    Let $S\leq \F^n$ be a codimension-$1$ subspace. Then $S$ is \vgood if and only if $S = \{\vec{e}_i\}^\perp$ for some $1 \leq i \leq n$.
\end{restatable}

\begin{proof}
We prove the ``only if'' direction by induction on $n$. To avoid confusion, we use superscripts to denote the dimension of a vector, e.g., $\vec{e}_1^2 = \tr{(1, 0)}$ and $\vec{e}_1^3 = \tr{(1, 0, 0)}$.

The claim is trivial for $n=1$. For $n=2$, we show that if $\{\vec{v}\}^\perp$ is \vgood, then $\vec{v}$ is a scalar multiple of $\vec{e}_1^2$ or $\vec{e}_2^2$. Indeed, the only nonzero vectors that are not a scalar multiple of either $\vec{e}_1^2$ or $\vec{e}_2^2$ are of the form $\tr{(\lambda, \mu)}$ with $\lambda, \mu\neq 0$. Note that
\begin{align*}
\per\begin{mx}
    \mu & \mu\\
    -\lambda & -\lambda
\end{mx} = -2\lambda\mu\neq 0
\end{align*}
since $\lambda, \mu\neq 0$ and the characteristic $p \neq 2$. Since this matrix has columns in $\{\vec{v}\}^\perp$, we conclude that if $S = \{\vec{v}\}^\perp$ is codimension-$1$ \vgood space, then $\vec{v} = \lambda\vec{e}^2_i$.

Suppose now that $n\geq 2$ and we have proved the claim for all values $1, 2, \ldots, n$; we shall prove it for $n+1$. Let $\vec{v} = \tr{(v_1, \ldots, v_{n+1})}$ be some nonzero vector in $\F^{n+1}$ not of the form $\lambda\vec{e}^{n+1}_i$, we prove that $\{\vec{v}\}^\perp$ is not \vgood by finding a matrix with columns in $\{\vec{v}\}^\perp$ and non-zero permanent. 

\begin{description}
    \item[Case 1] $\vec{v}$ has at least one zero coordinate. Assume without loss of generality $v_{n+1} = 0$ and denote $\widehat{\vec{v}} = \tr{(v_1, \ldots, v_n)}$. Note that $\widehat{\vec{v}}$ is nonzero and not of the form $\lambda\vec{e}^n_i$. By the induction hypothesis, there exist $\widehat{\vec{r}}_1, \ldots, \widehat{\vec{r}}_n \in \{\widehat{\vec{v}}\}^\perp$ such that
    $$\per\begin{mx}
    \widehat{\vec{r}}_1&
    \cdots&
    \widehat{\vec{r}}_n
    \end{mx}\neq 0.$$
    Let $\vec{r}_i = \tr{(\widehat{\vec{r}}_i, 0)}$. Then $\vec{e}^{n+1}_{n+1}, \vec{r}_1, \vec{r}_2, \ldots, \vec{r}_n \in \{\vec{v}\}^\perp$ and
    $$
    \per\begin{mx}
        \vec{r}_1&
        \cdots&
        \vec{r}_n&
        \vec{e}_{n+1}
    \end{mx} = \per\begin{mx}
    \widehat{\vec{r}}_1&
    \cdots&
    \widehat{\vec{r}}_n
    \end{mx}\neq 0.
    $$
    This proves that $\{\vec{v}\}^\perp$ is not \vgood.

    \item[Case 2] $\vec{v} \in \spn\{\tr{(1, 1, \ldots, 1)}\}$. We prove directly that
    \begin{align*}
    &&A_n = 
    \begin{mx}
        1 & 0 & 0 & \cdots & 0 & 0 \\
        -1 & 1 & 0 & \cdots & 0 & 0 \\
        0 & -1 & 1 & \cdots & 0 & 0 \\
        \vdots & \vdots & \vdots & \ddots & \vdots &\vdots \\
        0 & 0 & 0 & \cdots & 1 & 1 \\
        0 & 0 & 0 & \cdots & -1 & -1
    \end{mx}
    \end{align*}
    (with columns in $\{\vec{v}\}^\perp$) has a nonvanishing permanent. The proof proceeds by induction:
    $$\per\begin{mx}
    1 & 1 \\
    -1 & -1
    \end{mx} = -2\neq 0$$
    and for $n\geq 3$ we simply expand $A_n$ by the first row to obtain $\per A_n = \per A_{n-1}$.

    \item[Case 3] None of the coordinates of $\vec{v}$ is $0$ and $\vec{v}\notin\spn\{\tr{(1, 1, \ldots, 1)}\}$. In this case, there are two coordinates $(i, j)$ such that $(v_i, v_j) = (\lambda, \mu)$ for some distinct nonzero scalars $\lambda\neq\mu$, and we assume without loss of generality that $(v_n, v_{n+1}) = (\lambda, \mu)$. 
    \begin{description}
        \item[Case 3a $n+1=3$] In this case $\vec{v} = \tr{(\alpha, \lambda, \mu)}$ for some nonzero scalar $\alpha$.
        \begin{align*}
        &\per\begin{mx}
            \lambda & 0 & 0 \\
            -\alpha & \mu & \mu \\
            0 & -\lambda & -\lambda
        \end{mx} = -2\mu\lambda^2\neq 0
        \end{align*}
        so that $\{\vec{v}\}^\perp$ is not \vgood.

        \item[Case 3b $n+1\geq 4$] In this case we denote $\widehat{\vec{v}} = \tr{(v_1, v_2, \ldots, v_{n-1})}$ (dropping the last two coordinates of $\vec{v}$) and note that it is a nonzero vector not of the form $\lambda\vec{e}^{n-1}_i$ (because it has at least two nonzero coordinates). By the induction hypothesis there exist $\widehat{\vec{r}}_1, \ldots, \widehat{\vec{r}}_{n-1} \in \{\widehat{\vec{v}}\}^\perp$ such that
        $$\per\begin{mx}
            \widehat{\vec{r}}_1&
            \widehat{\vec{r}}_2&
            \cdots &
            \widehat{\vec{r}}_{n-1}
        \end{mx}\neq 0.$$
        Let $\vec{r}_i = \tr{(\widehat{\vec{r}}_i, 0, 0)}$ and $\vec{s} = \tr{(0, 0, \ldots, 0, \lambda, -\mu)}$. Then $\vec{s}, \vec{r}_1, \ldots, \vec{r}_{n-1} \in (\spn\{\vec{v}\})^{\perp}$ and
        \begin{align*}
            \per&\begin{mx}
                \vec{r}_1&
                \vec{r}_2&
                \cdots&
                \vec{r}_{n-1}&
                \vec{s}&
                \vec{s}
            \end{mx} 
        =\\ 
        &-2\mu\lambda\per\begin{mx}
            \widehat{\vec{r}}_1&
            \widehat{\vec{r}}_2&
            \cdots &
            \widehat{\vec{r}}_{n-1}
        \end{mx}\neq 0.
        \end{align*}
    \end{description}
\end{description}
\end{proof}

\subsection{The base case of the induction}
In this subsection we prove Lemma \ref{lem:3friends}, the base case ($n=3$) of the classification of \frnd subspaces. We proceed via a sequence of lemmas.

\begin{lem}[Classification of $n\times 2$ matrices]\label{lem:i2}
    Let $n\geq 2$. An $n\times 2$ matrix $A \in \F^{n\times 2}$ has non-full permanental rank $\prk{A}<2$ if and only if exactly one of these three possibilities holds:
    \begin{enumerate}[label=(\roman*)]
        \item\label{lemi2:c1} the matrix has a zero column;
        \item\label{lemi2:c2} there is a single nonzero row, and that row is a scalar multiple of $(1, x)$ (for some nonzero $x \in \F$);
        \item\label{lemi2:c3} there are exactly two nonzero rows: one a scalar multiple of $(1, x)$ and the other a scalar multiple of $(1, -x)$ (for some nonzero $x \in \F$).
    \end{enumerate}
\end{lem}

\begin{proof}
Let $A$ be an $n\times 2$ matrix with $\prk{A} < 2$. Each row of $A$ is a scalar multiple of one of the $\crd{\F}+2$ vectors
\begin{align*}
    &(0, 0) &&(0, 1) &&(1, x)
\end{align*}
where $x \in \F$ is an arbitrary field element. Which of these vectors may appear together as rows in $A$? First, consider vectors of the form $(1, x)$ with $x \neq 0$. Since
$$\per\begin{mx}1 & x \\ 1 & y\end{mx} = x+y = 0 \iff y = -x$$
we see that $(1, x)$ and $(1, -x)$ may appear in $A$ together or individually, each appearing at most once. Since
$$\per\begin{mx}1 & x \\ 0 & 1\end{mx} = 1$$
no vector of the form $(1, x)$ may appear with $(1, 0)$. 
\end{proof}

Next we need some basic facts about codimension-$1$ spaces in $\F_q^3$. We shall consider a function $\Per_{3\times 2}: \F_q^{3\times 2}\to \F_q^3$
\begin{align*}
    \begin{mx}
        a_{11} & a_{12} \\
        a_{21} & a_{22} \\
        a_{31} & a_{32}
    \end{mx} \mapsto \begin{mx}
        \per \begin{mx}
            a_{21} & a_{22}\\
            a_{31} & a_{32}
        \end{mx} &
        \per \begin{mx}
            a_{11} & a_{12} \\
            a_{31} & a_{32}
        \end{mx} &
        \per \begin{mx}
            a_{11} & a_{12}\\
            a_{21} & a_{22}
        \end{mx}
    \end{mx}^T
\end{align*}
(Note that the order of the subpermanents corresponds to expansion by the first column of a $3\times 3$ permanent.) Given some $U, V \leq \F_q^3$ subspaces of dimension $\geq 2$, we are interested in the dimension of
$$P(U, V) \eqdef \spn\stt{\Per_{3\times 2}\begin{mx}\vec{u}& \vec{v}\end{mx}}{\vec{u} \in U, \vec{v} \in V}.$$
Lemma \ref{lem:i2} implies that $\dim P(U, V) > 0$.

\begin{lem}\label{lem:1Per}
    Let $U, V \leq \F_q^3$ be subspaces, each of dimension $\geq 2$. Then $\dim P(U, V) = 1$ if and only if $U = V = \{\vec{e}_i\}^\perp$ for some $1 \leq i \leq 3$.
\end{lem}

\begin{proof}
    The ``if'' direction is immediate: if $U = V = \{\vec{e}_i\}^\perp$, then it is easy to see that $P(U, V) = \spn\{\vec{e}_i\}$. For the ``only if'' direction, we assume that $\dim P(U, V) = 1$ and start by noting that it is not possible for either $U$ or $V$ to be full dimensional. Indeed, assume without loss of generality $V = \F_q^3$. Choosing any two linearly independent vectors $\vec{a}, \vec{b} \in U$ we know that $\begin{mx} \vec{a} & \vec{b}\end{mx}$ has rank $2$, so there exists some coordinate $1 \leq i\leq 3$ such that the vectors $\widehat{\vec{a}}, \widehat{\vec{b}} \in \F_q^2$ obtained from $\vec{a}, \vec{b}$ by deleting the $i$-th coordinate are linearly independent. For the purpose of illustration, say that $i=3$. Then,
    \begin{align*}
        &\Per_{3\times 2}\begin{mx}\vec{a}&\vec{e}_3\end{mx} = \tr{\begin{mx}a_2 & a_1 & 0\end{mx}}
        &&\Per_{3\times 2}\begin{mx}\vec{b}&\vec{e}_3\end{mx} = \tr{\begin{mx}b_2 & b_1& 0 \end{mx}}
    \end{align*}
    so that $\dim P(U, V) \geq 2$.
    
    Therefore, we may assume that $\dim U = \dim V = 2$. Since $\dim P(U, V) = 1$, we have $P(U, V) = \spn\{\vec{p}\}$ for some nonzero vector $\vec{p} = \tr{(p_1, p_2, p_3)} \in \F_q^3$.

    \begin{claim}
    	If $\vec{e}_i$ is an element of $U$ or $V$ (for some $1 \leq i \leq 3$) then $U = V = \{\vec{e}_j\}^\perp$ for some $1 \leq j \leq 3$.
    \end{claim}

    \begin{proof}[Proof of claim]
    	Without loss of generality, suppose $\vec{e}_i \in U$, and for concreteness suppose $\vec{e}_3\in U$ (the other cases are analogous).

        For every $\vec{v} = \tr{(v_1, v_2, v_3)} \in V$ there exists some scalar $\lambda \in \F$ such that we have
        	\[\Per_{3\times 2}\begin{mx}\vec{v} & \vec{e}_3\end{mx} = \tr{\begin{mx}v_2 & v_1 & 0\end{mx}} = \lambda\vec{p}.\]
        In particular, $p_3 = 0$, and $V = \spn\left\{\vec{e}_3, \begin{mx}p_2 & p_1 & 0\end{mx}\right\}$. But then $\vec{e}_3 \in V$ and the same reasoning shows that $U = \spn\left\{\vec{e}_3, \begin{mx}p_2 & p_1 & 0\end{mx}\right\} = V$. Consider now
        \begin{align*}
            \Per_{3\times 2}\begin{mx}p_2 & p_2 \\ p_1 & p_1\\ 0 & 0\end{mx} = \tr{\begin{mx}0 & 0 & 2p_1p_2\end{mx}} \in \spn\{\vec{p}\} \iff
             p_1p_2 = 0.
        \end{align*}
        If $p_1 = 0$ we have $U = V = \{\vec{e}_2\}^\perp$, whereas if $p_2 = 0$ we have $U = V = \{\vec{e}_1\}^\perp$.
    \end{proof}

    Therefore, it remains to show that $\vec{e}_i$ is an element of $U$ or $V$ (for some $1 \leq i \leq 3$). Suppose towards contradiction that for any $1 \leq i \leq 3$, $\vec{e}_i \notin U, V$. We start by finding nonzero vectors $\vec{u} \in U$ and $\vec{v} \in V$ such that $\Per_{3\times 2}\begin{mx}\vec{u}& \vec{v}\end{mx} = 0$.

    To find such vectors, fix $\vec{u}, \vec{u}'$ a basis for $U$ and some nonzero $V$-vector $\vec{0}\neq \vec{v}\in V$. The map $H: \F_q^3\times \F_q^3 \to \F_q^3$ taking a pair of vectors to their $\Per_{3\times 2}$ function $(\vec{u}, \vec{v}) \mapsto \Per_{3\times 2}\begin{mx}\vec{u} &\vec{v}\end{mx}$ is bilinear. We have
        \begin{align*}
            &H(\vec{u}, \vec{v}) = \lambda_{uv}\vec{p}
            &&H(\vec{u}', \vec{v}) = \lambda_{u'v}\vec{p}.
        \end{align*}
    For some scalars $\lambda_{uv}, \lambda_{u'v}$. If one of $\lambda_{uv}, \lambda_{u'v}$ is zero, we are done. Otherwise, bilinearity gives
    	\[H(\lambda_{uv}^{-1}\lambda_{u'v}\vec{u}-\vec{u}', \vec{v}) = \vec{0},\]
    and we know that $\lambda_{uv}^{-1}\lambda_{u'v}\vec{u}-\vec{u}' \neq \vec{0}$ since $\vec{u}, \vec{u}'$ are linearly independent. We conclude that it is possible to find some nonzero vectors $\vec{u} \in U$ and $\vec{v} \in V$ such that $\Per_{3\times 2}\begin{mx}\vec{u}& \vec{v}\end{mx} = 0$.
        
    Next, we apply the conclusion of Lemma \ref{lem:i2} to the matrix $\begin{mx}\vec{u} & \vec{v}\end{mx}$. Since $\vec{e}_i \notin U, V$ for any $1 \leq i \leq 3$, we must have 
    \[\begin{mx}\vec{u} & \vec{v}\end{mx} = \begin{mx}0 & 0 \\ a & ax \\ b & -bx\end{mx}\]
    for some nonzero scalars $a, b, x \neq 0$, up to a permutation of the rows.

    In particular, up to some permutation of the coordinates, we may choose $\vec{u} = \tr{(0, 1, c)} \in U$ and $\vec{v} = \tr{(0, 1, -c)} \in V$ for some $c \neq 0$. Since $U$ is a $2$-dimensional space and $\vec{e}_i \notin U$, we know that there exists a vector of the form $\tr{(1, 0, d)} \in U$ (with $d \neq 0$). Similarly, there exists a vector of the form $\tr{(1, 0, e)} \in V$ (with $e \neq 0$). We now observe that 
        \begin{align*}
            &\Per_{3\times 2}\begin{mx}0 & 1 \\ 1 & 0\\ a & c\end{mx} = \tr{\begin{mx}c & a & 1\end{mx}}
            &&\Per_{3\times 2}\begin{mx}1 & 0\\ 0 & 1\\ b & -a \end{mx} = \tr{\begin{mx}b & -a & 1\end{mx}}
        \end{align*}
    are both elements of $\spn\{\vec{p}\}$, and so are nonzero scalar multiples of each other. Now,
    	\[\tr{\begin{mx}c & a & 1\end{mx}} = \lambda\tr{\begin{mx}b & -a & 1\end{mx}}\]
    if and only if $\lambda = 1$, $a = 0$, and $b=c$. In particular, $U\ni \tr{\begin{mx}0 & 1 & a\end{mx}} = \vec{e}_2$, contradicting the assumption that $\vec{e}_i \notin U, V$ for any $1 \leq i \leq 3$.
\end{proof}

The proof of the base case of Theorem \ref{thm:manyfriends} follows easily from Lemma \ref{lem:1Per}.

\begin{lem}\label{lem:3friends}
    Let $S_1, S_2, S_3 \leq \F_q^3$ be three subspaces, each of dimension at least $2$. Then $S_1, S_2, S_3$ is a \frnd list if and only if  $S_1 = S_2 = S_3 = \{\vec{e}_i\}^{\perp}$ for some $1 \leq i \leq 3$.
\end{lem}

\begin{proof}
    If the three subspaces are identical, it is clear they cannot be of full dimension, and so the claim follows from the classification of codimension-$1$ \vgood subspaces, Theorem \ref{thm:c1vgood}. Assume for contradiction the three spaces are not identical, and without loss of generality that $S_3$ is different from both $S_1$ and $S_2$ (while $S_1$ and $S_2$ may or may not be the same subspace).

    Since $S_1, S_2, S_3$ are \frnd, for any choice of $\vec{x} \in S_1$, $\vec{u} \in S_2$ and $\vec{v} \in S_3$, we have
    $$\per\begin{mx}\vec{x}&\vec{u}&\vec{v}\end{mx} = 0.$$
    In other words, $\vec{x} \in \left\{\Per_{3\times 2}\begin{mx}\vec{u} & \vec{v}\end{mx}\right\}^\perp$, or more generally $S_1 \perp P(S_2, S_3)$. This is a contradiction since $\dim S_1 \geq 2$ and by Lemma \ref{lem:1Per}, $\dim P(S_2, S_3) \geq 2$.
\end{proof}

\subsection{A graph-theoretic lemma}
The inductive step of the proof of the classification of \frnd subspaces (Theorem \ref{thm:manyfriends}) relies on identifying a suitable coordinate to delete, so the inductive hypothesis can be invoked for the resulting vectors. In order to guarantee the existence of such a coordinate we need a graph-theoretic lemma. 

\begin{defi}
    Let $G = (V, E)$ be a bipartite graph with partition $V = A \dunion B$. We say $G$ is a \emph{$B$-star} if $B$ has a unique vertex of positive degree---called the \emph{centre} of the star. (See Figure \ref{fig:graphs}.)
\end{defi}

We now prove that if the original graph $G$ has no isolated vertices, then it is possible to delete an $A$-vertex and a $B$-vertex such that the resulting induced graph is not a $B$-star.

\begin{restatable}[Star-avoidance lemma]{lem}{nostar}
\label{lem:nostar}
    Let $G = (V, E)$ be a bipartite graph with partition $V = A \dunion B$ such that
\begin{description}
    \item[(P1)] $\crd{A}, \crd{B} \geq 4$;
    \item[(P2)] every $a \in A$ has a positive degree;
    \item[(P3)] every $b \in B$ has a positive degree.
\end{description}
Then it is possible to delete two vertices, $a \in A$ and $b \in B$, such that:
\begin{description}
    \item[(R1)] $b$ is not the unique neighbour of $a$ in $G$;
    \item[(R2)] after deleting $a$ and $b$, the resulting induced graph $G'$ is neither empty nor a $B$-star.
\end{description}
\end{restatable}

\mbox{}

\begin{figure}[H]
\centering
\includegraphics[scale=0.3]{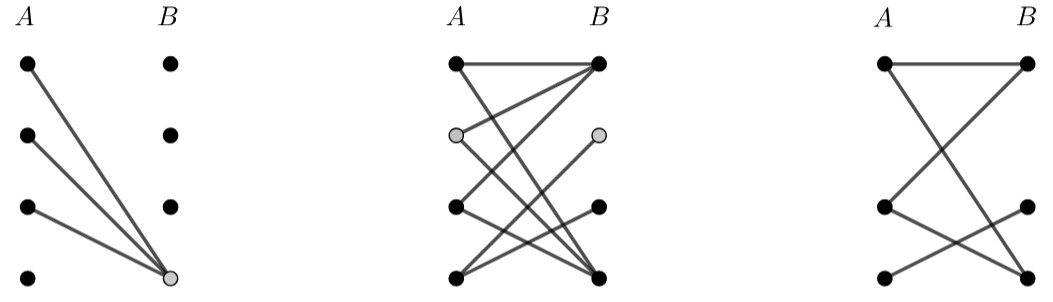}
\caption{\tiny (Left) Example of a $B$-star, the lighter vertex is the centre.
(Middle) A bipartite graph; the two lighter vertices are the ones marked for removal in the proof below. Note that removing the top $B$-vertex and the bottom $A$-vertex creates the star on the left.
(Right) The induced subgraph resulting from removing the two lighter vertices of the graph in the middle. }
\label{fig:graphs}
\end{figure}

\mbox{}

\begin{proof}
    Note that if $G$ is a full bipartite graph, then we may take any vertex in $B$ for $b$ and any vertex in $A$ for $a$; so we may assume that $G$ is not full.
    
    We now claim that we may take $b$ to be a minimum degree vertex in $B$ and $a$ a minimum degree vertex outside the neighbourhood $\Gamma(b)$ of $b$, i.e.~a minimum degree vertex among $\stt{a \in A}{ab\notin E}$. Clearly, the choice of $b$ and $a$ satisfies condition (R1); it remains to show that it also satisfies condition (R2).

    First note that $G'$ is not empty, for otherwise it means that every vertex in $A$ other than $a$ has $b$ as its unique neighbour, and this contradicts the assumption that $b$ is a minimal-degree vertex in $B$.

    Next, assume for contradiction that $G'$ is a $B$-star with centre $r$. Since $\crd{B}\geq 4$, we know that in addition to $b$ and $r$ there are at least two more vertices $s, t \in B$. Likewise, since $\crd{A}\geq 4$, in addition to $a$, there are at least three more vertices $x, y, z \in A$.

    Since $G'$ is a $B$-star with centre $r$, we have $\deg_{G'}(s) = \deg_{G'}(t) = 0$. Since by (P2) $\deg_G(s), \deg_G(t) > 0$, we see that $as, at \in E$, so that $\deg_G(a) \geq 2$. By the choice of $a$, every vertex $a' \in A$ not connected to $b$ has $\deg_G(a')\geq 2$. In particular, if $x, y, z$ are not connected to $b$, then $2\leq \deg_G(x) = \deg_{G'}(x)$ and similarly for $y, z$. However, since $G'$ is a star, we have $\deg_{G'}(x) = \deg_{G'}(y) = \deg_{G'}(z) \leq 1$. This proves that $xb, yb, zb \in E$.

    We conclude that $\deg_G(b) \geq 3$. By the choice of $b$, every vertex $b' \in B$ has $\deg_G(b') \geq 3$. In particular, $\deg_{G}(s) \geq 3$. We now have the contradiction $0 = \deg_{G'}(s) \geq \deg_G(s)-1 = 2$.
\end{proof}

\subsection{Putting everything together}
Having completed the necessary groundwork, we are ready to prove our main structural result.

\manyfriends

\begin{proof}
  By induction on $n$. The base case $n=3$ is Lemma \ref{lem:3friends}. Suppose we have proved the theorem for some $n\geq 3$ and we wish to prove it for $n+1$. Towards that end, let $S_1, \ldots, S_{n+1} \leq \F_q^{n+1}$ be a \frnd list of subspaces, each of dimension $\geq n$ and assume for contradiction that they are not all of the form $\{\vec{e}_i\}^\perp$ for the same $i$. To derive a contradiction, we show how to construct an $(n+1)\times(n+1)$ matrix $A$ whose $j$-th column is an element of $S_j$ such that $\per A \neq 0$, contradicting the assumption that $S_1, \ldots, S_{n+1}$ is a \frnd list.

  We may assume without loss of generality that each $S_j$ is of codimension $1$ (otherwise, $S_j = \F_q^{n+1}$ and we replace it with a subspace of codimension $1$). In particular, we may assume that each $S_j = \{\vec{v}_j\}^\perp$ for some nonzero $\vec{v}_j\in \F_q^{n+1}$, and we may choose each $\vec{v}_j$ such that its first nonzero coordinate is $1$. By hypothesis, there is no $1 \leq i \leq n+1$ such that all the $\vec{v}_j$ equal $\vec{e}_i$.

  \begin{description}
    \item[Case] Consider first the possibility that all the $\vec{v}_j$ have some common coordinate that is $0$; say the last coordinate. Let $\widehat{\vec{v}}_j \in \F_q^n$ be the vector $\vec{v}_j$ with the last coordinate removed, and let $\widehat{S}_j = \{\widehat{\vec{v}}_j\}^\perp$. Then $\widehat{S}_1, \ldots, \widehat{S}_n$ is a \frnd list.

    To see this, suppose $\widehat{A}$ is a matrix with $j$-th column $\widehat{\vec{a}}_j \in \widehat{S}_j$, for $1 \leq j \leq n$. We build a matrix $A$ whose $j$-th column is in $S_j$ and such that $\per \widehat{A} = \per A$ which vanishes by the hypothesis that $S_1, \ldots, S_{n+1}$ is a \frnd list. For $1 \leq j \leq n$, the $j$-th column of $A$ is $\tr{(\widehat{\vec{a}}_j, 0)}$, and the $(n+1)$-st column of $A$ is $\vec{e}_{n+1}$.

    This proves that $\widehat{S}_1, \ldots, \widehat{S}_n$ is a \frnd list, so the induction hypothesis implies that there exists some $1 \leq i \leq n$ such that $\widehat{\vec{v}}_j = \vec{e}_i \in \F_q^n$ for all $1 \leq j \leq n$. 

    Repeating the same argument with $\widehat{S}_2, \ldots, \widehat{S}_{n+1}$ we find that $\widehat{\vec{v}}_{n+1} = \vec{e}_i$. Since the last coordinate of $\vec{v}_j$ is $0$ (for $1 \leq j \leq n+1$), we see that $\vec{v}_j = \vec{e}_{i}\in \F_q^{n+1}$, contradicting the assumption that there is no $1 \leq i \leq n+1$ such that all the $S_j$ equal $\{\vec{e}_i\}^\perp$.

    \item[Case] the $\vec{v}_j$'s do not have a common coordinate which is $0$. In other words, for every coordinate $1 \leq i \leq n+1$ there is some $\vec{v}_j$ whose $i$-th coordinate is nonzero. 

    Consider the bipartite graph $G = (V, E)$ with partition $V = A \dunion B$ where $A \eqdef \{\vec{v}_1, \ldots, \vec{v}_{n+1}\}$ is the set of vectors and $B \eqdef \{1, 2, \ldots, n+1\}$ the set of coordinates; a vector is connected by an edge to its support, i.e.~$\vec{v}_jk \in E$, if and only if the $k$-th coordinate of $\vec{v}_j$ is nonzero. Note that the three premises of the Star-avoidance Lemma (Lemma \ref{lem:nostar}) are satisfied:
  \begin{description}
    \item[(P1)] $\crd{A}, \crd{B} \geq 4$, since $n+1\geq 4$;
    \item[(P2)] every $a \in A$ has a positive degree, since the vectors $\vec{v}_j$ are all nonzero;
    \item[(P3)] every $b \in B$ has a positive degree, since the vectors $\vec{v}_j$ do not have a common coordinate zero.
  \end{description}

  By the Star-avoidance Lemma it is possible to find a coordinate $b \in B$ and a vector $\vec{v} \in A$ such that 
    \begin{description}
        \item[(R1)] $b$ is not the unique neighbour of $a$ in $G$, i.e.~$\vec{v} \neq \vec{e}_b$;
        \item[(R2)] after deleting $a$ and $b$, the resulting induced graph $G'$ is neither empty nor a $B$-star. In other words, letting $\widehat{\vec{v}}_j \in \F_q^n$ denote the vector resulting from $\vec{v}_j$ by deleting the $b$-th coordinate, there is no $1 \leq i \leq n$ such that all of the $\widehat{\vec{v}}_j$ (not counting $\widehat{\vec{v}}$) are in the span of $\vec{e}_i^n$.
    \end{description}
    For ease of writing, we shall assume that $\vec{v} = \vec{v}_{n+1}$ and $b = n+1$.
    
    Letting $\widehat{S}_j \eqdef \{\widehat{\vec{v}}_j\}^\perp$ (for $1\leq j\leq n$), the induction hypothesis together with (R2) imply that the $\widehat{S}_1, \ldots, \widehat{S}_n$ is \emph{not} a \frnd list. In particular, it is possible to find vectors $\widehat{\vec{a}}_j \in \widehat{S}_j$ (for $1 \leq j\leq n$) such that the matrix $\widehat{A}$ whose $j$-th column is $\widehat{\vec{a}}_j$ ($1 \leq j \leq n$) has $\per \widehat{A} \neq 0$. We shall now construct a matrix $A$ whose $j$-th column $\vec{a}_j$ in $S_j$ ($1 \leq j \leq n+1$) and such that $\per A = \per \widehat{A}$, contradicting the assumption that $S_1, \ldots, S_{n+1}$ is a \frnd list.
    
    For $1 \leq j \leq n$ we define $\vec{a}_j\eqdef (\widehat{\vec{a}}_j, 0)$ (for $1 \leq j \leq n$) and note that $\vec{a}_j \in S_j$. For $\vec{a}_{n+1}$ we take any vector in $S_{n+1}$ with $1$ in the $b$-th coordinate; note that this is possible since by (R1) we know that $\vec{v}\neq \vec{e}_b$. Expanding by the $b$-th row we have $\per A = \per \widehat{A} \neq 0$.
  \end{description}
\end{proof}

We would like to draw the reader's attention to the fact that the only property of the permanent that is used in the proof above is the cofactor expansion formula. As this formula is shared with the determinant, the inductive step would have also worked for the determinant, even though the result is false for the determinant --- as we remarked in the introduction, it is the base case that makes all the difference.

\section{\vgood subspaces in sufficiently large characteristic}\label{sec:0fld}
Jointly-permanull subspaces of codimension $1$ were a crucial component in our proof of Theorem \ref{thm:main}. Understanding \frnd subspaces of larger codimension looks like a very interesting problem, and seems relevant to extending the range of $k$ in that theorem.

In this section and the next, we make some small steps in this direction. We begin with a very simple proof showing that for fields of characteristic $0$ or $> n$, the only \vgood subspaces of $\F^n$ are the trivial ones, i.e.~subspaces of $\{\vec{e}_i\}^\perp$ for some $i$ (Theorem \ref{thm:simpletrivial}). 

With a more refined argument, we show how to do better for \vgood subspaces of bounded codimension; we prove that this characterization continues to hold with the weaker hypothesis $\chr{\F} > k+1$ where $k$ is the codimension of the \vgood subspace, for any ambient dimension $n$ (Theorem \ref{thm:trivialpermanull}). This generalizes Theorem~\ref{thm:c1vgood}, the characterization for the $k = 1$ case, where the only hypothesis was that $\chr{\F} > 2$.

Once $\chr{\F} \leq k+1$, there can be nontrivial \vgood spaces of codimension $k$. In Section \ref{sec:poly} we go beyond the $\chr{\F}>k+1$ restriction and give a characterization of \vgood subspaces in terms of the vanishing of the \emph{permanental polynomial} (Theorem \ref{thm:zeropoly}). This also allows gives a computational criterion for determining whether a given subspace is \vgood. We conclude the paper by showing how to use this polynomial to give an alternative, simplified proof of Theorem \ref{thm:trivialpermanull}.

\subsection{A simple proof for $\chr{\F} > n$}

We start with the following observation, which is an immediate consequence of the definition of \vgood subspaces.
\begin{lem}\label{lem:downclosed}
    Any subspace of a \vgood subspace is a \vgood subspace.
\end{lem}

\begin{proof}
	Immediate from the definition of \vgood subspaces.
\end{proof}

The necessity of some assumption on the characteristic is apparent in the proof of the following easy lemma.

\begin{lem}\label{lem:dim1vgood}
    Let $V \leq \F^n$ be a \vgood subspace of dimension $1$. If $\F$ is of characteristic $0$ or $\chr\F > n$, then $V\subseteq \{\vec{e}_i\}^\perp$ for some $1 \leq i \leq n$.
\end{lem}

\begin{proof}
    Any vector $\vec{v}$ all of whose coordinates are nonzero induces a matrix
    \[\begin{mx}
        \vec{v} & \vec{v} & \cdots & \vec{v}
    \end{mx}\]
    with nonzero permanent. Indeed, if $\vec{v} = (v_1, \ldots, v_n)$, then by multilinearity
    \[\per\begin{mx}
        \vec{v} & \vec{v} & \cdots & \vec{v}
    \end{mx} = v_1\cdots v_n\cdot \per\begin{mx}
        \vec{1} & \vec{1} & \cdots & \vec{1}
    \end{mx} = v_1\cdots v_n\cdot n!\]
    (where $\vec{1}$ is the vector all of whose entries are $1$). Since $v_1, \ldots, v_n$ are nonzero, this product is nonzero as long as $n! \neq 0$, i.e.~as long as $\chr\F = 0$ or $\chr\F \geq n+1$.
\end{proof}

\begin{crl}\label{crl:0coord}
    Let $V \leq \F^n$ be a \vgood subspace. If $\F$ is of characteristic $0$ or $\chr\F > n$, then every $\vec{v} \in V$ has at least one zero coordinate.
\end{crl}

\begin{proof}
    By Lemma \ref{lem:downclosed}, $\spn\{\vec{v}\}$ is a \vgood subspace so by Lemma \ref{lem:dim1vgood} it must lie within $\{\vec{e}_i\}^\perp$ for some $1 \leq i \leq n$.
\end{proof}

Using this corollary, we may prove directly that for fields of characteristic zero (or sufficiently large characteristic) the only \vgood spaces are the trivial ones.

\begin{thm}
    Let $V \leq \F^n$ be a \vgood space. If $\F$ is of characteristic $0$ then $V \subseteq \{\vec{e}_i\}^{\perp}$ for some $1 \leq i \leq n$.
\end{thm}

\begin{proof}
    Assume for contradiction that $V\not\subseteq \{\vec{e}_i\}^\perp$ for any $1 \leq i \leq n$. Fix some basis $\vec{v}_1, \ldots, \vec{v}_\ell$ for $V$; then for any coordinate $1 \leq i \leq n$ at least one of these basis vectors has a nonzero $i$-th coordinate. Consider the $n$ polynomials defined by the $n$ coordinates of the vector
    \[\vec{v}_1 + x\vec{v}_2 + x^2\vec{v}_3 + \cdots + x^{n-1}\vec{v}_n.\]
    Let $Z$ be the set of field elements for which at least one of the above polynomials vanishes. Since each polynomial equation in this system has at most $n-1$ roots in $\F$, we have $\crd{Z} \leq n(n-1)$. Therefore, if $\crd{\F}>n(n-1)$ there is some $x$ for which none of these polynomials vanishes, i.e., 
    \[\vec{v}_1 + x\vec{v}_2 + \cdots + x^{n-1}\vec{v}_n\]
    is a vector of the \vgood subspace $V$ none of whose entries is $0$, contradicting Corollary \ref{crl:0coord}.
\end{proof}

Let us focus again on finite fields. The preceding proof shows that if $\chr\F > n$ (a necessary hypothesis for Corollary \ref{crl:0coord}) and $\crd{\F}>n(n-1)$, then any \vgood subspace of $\F^n$ is trivial. We can circumvent the restriction on the field size $\crd{\F}$ via a straightforward application of the probabilistic method. We make use of the following simple observation.

\begin{lem}
    Let $\F$ be a finite field and let $x_1, \ldots, x_\ell \in \F$ be some fixed field elements, at least one of which is nonzero. If $\alpha_1, \ldots, \alpha_\ell \in \F$ are selected uniformly at random, then
    \[x:= \alpha_1x_1 + \cdots + \alpha_\ell x_\ell\]
    is a uniformly random field element; that is, for every $\alpha \in \F$, $\Pr[x = \alpha] = 1/\crd{\F}$.
\end{lem}

\begin{proof}
    Let $d$ be the number of nonzero elements among $x_1, \ldots, x_\ell$. Then 
    \[\stt{(\alpha_1, \ldots, \alpha_\ell)\in \F^\ell}{\sum_{i=1}^{\ell}\alpha_ix_i = 0}\]
    is an $(\ell-d)$-dimensional linear subspace of $\F^\ell$. Moreover, for any $\alpha \in \F$
    \[\stt{(\alpha_1, \ldots, \alpha_\ell)\in\F^\ell}{\sum_{i=1}^{\ell}\alpha_ix_i = \alpha}\]
    is a coset of that subspace. In particular, all of these sets have the same size.
\end{proof}

\begin{thm}\label{thm:simpletrivial}
    Let $V \leq \F^n$ be a \vgood space. If $\chr\F>n$, then $V \subseteq \{\vec{e}_i\}^{\perp}$ for some $1 \leq i \leq n$.
\end{thm}

\begin{proof}
    Assume for contradiction that $V\not\subseteq \{\vec{e}_i\}^\perp$ for any $1 \leq i \leq n$. Fix some basis $\vec{v}_1, \ldots, \vec{v}_\ell$ for $V$; then for any coordinate $1 \leq i \leq n$ at least one of these basis vectors has a nonzero $i$-th coordinate. If the coefficients of the linear combination
    \[\vec{v}:= \alpha_1\vec{v}_1 + \cdots + \alpha_n\vec{v}_n\]
    are chosen uniformly at random, the $i$-th coordinate of $\vec{v}$ is also uniformly random; in particular, it is $0$ with probability $1/{\crd{\F}}$. Using the union bound, the probability that at least one coordinate of $\vec{v}$ is $0$ is at most $n/{\crd{\F}}$. Since $\crd{\F} \geq \chr{\F}>n$, there exists a linear combination such that all of the coordinates of $\vec{v}$ are nonzero, contradicting Corollary \ref{crl:0coord}.
\end{proof}

\subsection{Improved results for bounded codimension}
We have just seen that under some circumstances, the only \vgood spaces are trivial. The proof of Lemma \ref{lem:dim1vgood} shows that $\chr{\F}>n$ is a necessary condition for \vgood subspaces of dimension $1$. On the other hand, Theorem \ref{thm:c1vgood} on \vgood subspaces of codimension $1$ required only that the characteristic is $> 2$, independent of $n$. This suggests that ``sufficiently large'' characteristic should depend on the codimension $n - \dim V$ rather than on the ambient dimension $n$ alone. Indeed, these two extremes of codimension $1$ and codimension $n-1$ suggest that a suitable bound would be $\chr\F > (n-\dim V)+1$. Our goal for the remainder of the section is to prove such a bound.

The proof proceeds by induction on the ambient dimension $n$. The base case of $n=1$ is trivial, whereas for $n=2$, any nontrivial subspace is of codimension $1$, and so trivial by Theorem \ref{thm:c1vgood}. For the inductive proof, we shall make use of the following claim, where we show that Option (ii) below never holds!

\begin{claim}\label{claim:intersect}
    Suppose some $n \geq 1$ has the following property: for every \vgood subspace $U\leq \F^n$, if $\chr\F>\codim U + 1$, then $U$ is trivial (i.e., a subspace of $\{\vec{e}_i\}^\perp$ for some $1 \leq i \leq n$). Then every \vgood subspace $V \leq \F^{n+1}$ has the following property: if $\chr\F>\codim V+1$ then either:
    \begin{enumerate}[label=(\roman*)]
      \item there exists some index $1 \leq i \leq n+1$ such that $V\subseteq \{\vec{e}_i\}^\perp$; or
      \item for every index $1 \leq i \leq n+1$ we have $\dim(V\cap \{\vec{e}_i\}^\perp) = \dim V-1$. Moreover, there exists a different index $1 \leq j \leq n+1$ (with $j \neq i)$ such that
        \[V\cap \{\vec{e}_i\}^\perp = V\cap \{\vec{e}_j\}^\perp.\]
    \end{enumerate}
\end{claim}

\begin{proof}
    Suppose $V\subseteq \F^{n+1}$ is a \vgood codimension-$k$ subspace. Note that if $k = 1$, then $V$ is trivial by Theorem \ref{thm:c1vgood}. Similarly, if $k = n$, then the conclusion is satisfied by Lemma \ref{lem:dim1vgood} above. We may therefore assume that $1 < k < n$. In particular, $d:= \dim V = n+1-k > 1$.

    Suppose Option (i) does not hold, i.e., for each index $1 \leq i \leq n+1$ there exists some $\vec{v} \in V$ whose $i$-th coordinate is nonzero. For convenience, we denote $V_i := V\cap \{\vec{e}_i\}^\perp$.

    First, note that since $V\not\subseteq \{\vec{e}_i\}^\perp$, we have $V+\{\vec{e}_i\}^\perp = \F^{n+1}$. Therefore, $\dim V_i = \dim V - 1$.

    Next, let us fix some index $1 \leq i \leq n+1$. Let $\vec{v}_1 \in V$ be a vector whose $i$-th coordinate is nonzero, and by replacing it with a scalar multiple if necessary, we may assume that the $i$-th coordinate is $1$. Using linear combinations, we may complete $\vec{v}_1$ to a basis of $V$ by choosing vectors $\vec{v}_2, \ldots, \vec{v}_{d}$ whose $i$-th coordinate is $0$. Let $\widehat{\vec{v}}_2, \ldots, \widehat{\vec{v}}_{d} \in \F^{n}$ be the vectors obtained from $\vec{v}_2, \ldots, \vec{v}_{d}$ (respectively) by deleting the $i$-th coordinate. We claim that $\widehat{V}:= \spn\{\widehat{\vec{v}}_2, \ldots, \widehat{\vec{v}}_{d}\}$ is a \vgood subspace of $\F^{n}$. This can be seen as follows: any $\widehat{\vec{u}} \in \spn\{\widehat{\vec{v}}_2, \ldots, \widehat{\vec{v}}_{d}\} \subseteq \F^{n}$ is a linear combination of the basis elements
    \[\widehat{\vec{u}} = \sum_{\ell=2}^{d}\alpha_\ell\widehat{\vec{v}}_\ell.\]
    By introducing a new zero coordinate, we obtain the corresponding vector $\vec{u} \in \spn\{\vec{v}_2, \ldots, \vec{v}_{d}\} \subseteq \F^{n+1}$:
    \[\vec{u} = \sum_{\ell=2}^{d}\alpha_\ell\vec{v}_\ell.\]
    If $\widehat{\vec{u}}_2, \ldots, \widehat{\vec{u}}_{n+1} \in \spn\{\widehat{\vec{v}}_2, \ldots, \widehat{\vec{v}}_{d}\}$ then
    \[\per\begin{mx}
        \vec{v}_1 & \vec{u}_2 & \cdots & \vec{u}_{n+1}
    \end{mx} = \per\begin{mx}
        \widehat{\vec{u}}_2 & \cdots & \widehat{\vec{u}}_{n+1}
    \end{mx}\]
    as can immediately be seen by expanding the $(n+1)\times (n+1)$ matrix on the left by its $i$-th row. Since $V = \spn\{\vec{v}_1, \ldots, \vec{v}_{d}\}$ is assumed to be a \vgood subspace, so must be $\widehat{V} = \spn\{\widehat{\vec{v}}_2, \ldots, \widehat{\vec{v}}_{k+1}\}$.

    Since $\widehat{V}$ is a codimension-$k$ \vgood subspace of $\F^n$, the hypothesis of the claim applies: if $\chr\F > k+1$ there is some index $1 \leq j' \leq n$ such the $j'$-th coordinate of $\widehat{\vec{v}}_2, \ldots, \widehat{\vec{v}}_{d}$ is $0$. This means there is some index $1 \leq j \leq n$ different from $i$ such that the $j$-th index of $\vec{v}_2, \ldots, \vec{v}_{d}$ is zero\footnote{More precisely, if $j' \leq i$ then $j = j'+1$ and otherwise $j' = j$.}. Now, $\dim V_i = \dim V_j = d-1$, and $\vec{v}_2, \ldots, \vec{v}_d$ is a list of linearly independent vectors belonging to each of these spaces. Therefore, 
        \[V\cap \{\vec{e}_i\}^\perp = V\cap \{\vec{e}_j\}^\perp\]
    so that Option (ii) holds for $V$, as claimed.
\end{proof}

We are now ready to prove the relaxation on the characteristic.

\begin{thm}\label{thm:trivialpermanull}
    Let $V \leq \F^n$ be a \vgood subspace of codimension $k$. If $\chr\F>k+1$, then $V\subseteq \{\vec{e}_i\}^\perp$ for some $1 \leq i \leq n$.
\end{thm}

\begin{proof}
    We proceed by induction on the ambient dimension $n$. The cases $n=1, 2$ are straightforward. Suppose the theorem holds for all dimensions up to and including some fixed $n$, and we shall prove it for $n+1$. Towards that end, let $V\leq \F^{n+1}$ be a \vgood subspace of codimension $k$ and suppose that $\chr\F>k+1$. Assume for contradiction that $V$ is not trivial; in particular, invoking Theorem \ref{thm:c1vgood} and Lemma \ref{lem:dim1vgood} we have $1<k<n$.

    Claim \ref{claim:intersect} applies with Option (ii) guaranteeing that for every index $1 \leq i \leq n+1$ we have a different index $1 \leq j \leq n+1$ (with $j \neq i)$ such that
        \[V\cap \{\vec{e}_i\}^\perp = V\cap \{\vec{e}_j\}^\perp.\]
    Let $I$ denote the set of indices which vanish together with the first:
        \[I := \stt{1 \leq j \leq n+1}{V\cap \{\vec{e}_1\}^\perp = V\cap \{\vec{e}_j\}^\perp}.\]
    Let $\iota := \crd{I}$. Without loss of generality, we may assume that $I$ consists of the first $\iota$ indices. For convenience, let us denote $V':= V\cap\{\vec{e}_1\}^\perp$. From Option (ii) we know that $\iota \geq 2$ and $\dim V' = \dim V-1 = n-k$. Moreover, since all vectors in this $(n-k)$-dimensional subspace $V'\leq \F^{n+1}$ have their first $\iota$ entries $0$, we must have $\iota \leq k+1$. We shall now derive a contradiction by showing that $\crd{I}>\iota$.

    Let $\vec{v}_1, \ldots, \vec{v}_\ell$ be a basis for $V'$ (so $\ell = n-k$). Let $\vec{u}$ complete this list to a basis for $V$. Since $\vec{u} \notin V'$, for each $i \in I$ the $i$-th coordinate of $\vec{u}$ is nonzero; for otherwise $\vec{u} \in V\cap\{\vec{e}_i\}^\perp = V'$. 

    For any $\vec{w} \in \F^{n+1}$, let $\widehat{\vec{w}} \in \F^{n+1-\iota}$ denote the projection of $\vec{w}$ on the last $n+1-\iota$ coordinates. We claim that $\spn\{\widehat{\vec{v}}_1, \ldots, \widehat{\vec{v}}_\ell\}\subseteq \F^{n+1-\iota}$ is a \vgood subspace. To see this, consider the embedding sending any vector $\widehat{\vec{w}} \in \F^{n+1-\iota}$ to the vector $\vec{w}\in \F^{n+1}$ obtained from $\widehat{\vec{w}}$ by inserting $\iota$ zeroes as the first $\iota$ coordinates (and then the entries of $\widehat{\vec{w}}$ for the remaining coordinates). If $\widehat{\vec{u}}_1, \ldots, \widehat{\vec{u}}_{n+1-\iota} \in \spn\{\widehat{\vec{v}}_1, \ldots, \widehat{\vec{v}}_\ell\}$ is any collection of vectors, then the $(n+1)\times (n+1)$ matrix
    \[\begin{mx}
        \vec{u} & \cdots & \vec{u} & \vec{u}_1 & \cdots & \vec{u}_{n+1-\iota}
    \end{mx}\]
    has columns in $V$ and therefore must have zero permanent. On the other hand, this is a block diagonal matrix; if $\vec{u} = (u_1, \ldots, u_{n+1})$ then the permanent is simply
    \[\per\begin{mx}
        \vec{u} & \cdots & \vec{u} & \vec{u}_1 & \cdots & \vec{u}_{n+1-i}
    \end{mx} = u_1\cdots u_i\cdot i!\cdot \per\begin{mx}
        \widehat{\vec{u}}_1 & \cdots & \widehat{\vec{u}}_{n+1-\iota}
    \end{mx}.\]
    Since by assumption $\chr\F > k+1 \geq \iota$, we see that $\iota! \neq 0$ so that
    \[\per\begin{mx}
        \widehat{\vec{u}}_1 & \cdots & \widehat{\vec{u}}_{n+1-\iota}
    \end{mx} = 0,\]
    proving that $\spn\{\widehat{\vec{v}}_1, \ldots, \widehat{\vec{v}}_\ell\}\subseteq \F^{n+1-\iota}$ is a \vgood subspace.

    Since $\iota\geq 1$, we may invoke the induction hypothesis to conclude that $\spn\{\widehat{\vec{v}}_1, \ldots, \widehat{\vec{v}}_\ell\}\subseteq \F^{n+1-\iota}$ must be trivial, so there exists some $1 \leq j \leq n+1-\iota$ such that $\spn\{\widehat{\vec{v}}_1, \ldots, \widehat{\vec{v}}_\ell\}\subseteq \{\vec{e}_j\}^\perp\subseteq\F^{n+1-\iota}$. It follows that
    \[V' = \spn\{\vec{v}_1, \ldots, \vec{v}_\ell\}\subseteq \{\vec{e}_{j+\iota}\}^\perp\subseteq \F^{n+1}\]
    and therefore that $j+\iota \in I$ (in addition to the first $\iota$ indices), contradicting the fact that $\crd{I} = \iota$.
\end{proof}

\section{\vgood subspaces in arbitrary characteristic}\label{sec:poly}

The goal of this section is to provide a complete characterization of \vgood subspaces, regardless of the characteristic of the field, in terms of an auxiliary polynomial which we call the \emph{permanental polynomial} (Theorem \ref{thm:zeropoly}). This allows us to check if a given codimension-$k$ subspace is \vgood with $O_k(n^k)$ field operations, and a general subspace with $2^{O(n)}$ field operations. In contrast, checking \vgood{ness} directly by fixing a basis $B = \{b_1, \ldots, b_{n-k}\}$ and checking the vanishing of all permanents of matrices with columns from $B$ takes $O((n-k)^n)$ time.

In Section \ref{sec:0fld} we have found that the only \vgood codimension-$k$ subspaces of $\F^n$ are trivial, provided that $\chr\F > k+1$ (Theorem \ref{thm:trivialpermanull}). Without the assumption on the characteristic, there are other possible \vgood subspaces, as the following example shows.

\begin{eg}[Nontrivial \vgood subspace]\label{eg:nontrivial}
    The codimension-$2$ subspace $S\leq \F_3^4$ spanned by
    \begin{align*}
        &\vec{v}_1 = \tr{\begin{mx}1 & 0 & 1 & 1\end{mx}}
        &&\vec{v}_2 = \tr{\begin{mx}0 & 1 & 1 & -1\end{mx}}
    \end{align*}
    is a \vgood subspace which is not contained in any \vgood subspace of codimension-$1$ (i.e., $S$ is non-trivial). Note that over $\F_{p^\ell}^4$ with $p>3$, $S$ is \emph{not} a \vgood subspace, as guaranteed by Theorem \ref{thm:trivialpermanull}; one way of seeing this is to compute, for example,
    \[\per \begin{mx}
        1 & 1 & 1 & 0\\
        0 & 0 & 0 & 1\\
        1 & 1 & 1 & 1\\
        1 & 1& 1 & -1
    \end{mx} = 6 \neq 0.\]
\end{eg}

In order to characterize \vgood subspaces in this general setting, we introduce an auxiliary algebraic object: the \emph{permanental polynomial.} To motivate the definition, let us start with codimension-$2$ \vgood subspaces and reprove that for fields of characteristic $p>3$, the only codimension-$2$ \vgood subspaces are the trivial ones; namely, subspaces of $\{\vec{e}_i\}^\perp$.

Suppose $S$ is a codimension-$2$ \vgood subspace. Permuting coordinates if necessary (note that the permanent is invariant under permutation of columns), we may fix a basis $\vec{v}_1, \ldots, \vec{v}_{n-2}$ for $S$ where 
\[\vec{v}_i = (\vec{e}_i^{n-2}, a_i, b_i)\]
for some scalars $a_i, b_i$.

Since $S$ is \vgood, for any $\vec{u}, \vec{w} \in S$, we have
\[\per\begin{mx}\vec{v}_1 & \cdots & \vec{v}_{n-2} & \vec{u} & \vec{w}\end{mx} = 0.\]
Expressing $\vec{u}, \vec{w}$ as linear combinations of the basis vectors:
\begin{align*}
    &\vec{u} = x_1\vec{v}_1 + \cdots + x_{n-2}\vec{v}_{n-2},\\
    &\vec{w} = y_1\vec{v}_1 + \cdots + y_{n-2}\vec{v}_{n-2},
\end{align*}
we may rewrite the matrix above as
\[\begin{mx}\vec{v}_1 & \cdots & \vec{v}_{n-2} & \vec{u} & \vec{w}\end{mx} = 
\begin{mx}
    & & & & & x_1 & y_1 \\
    & & & &  & x_2 & y_2 \\
    & & & I_{n-2} & &  x_3 & y_3 \\
    & & & & & \vdots & \vdots \\
    & & & & & x_{n-2} & y_{n-2} \\
    a_1 & a_2 & a_3 & \cdots & a_{n-2} & \tr{\vec{a}}\vec{x} & \tr{\vec{a}}\vec{y}\\
    b_1 & b_2 & b_3 & \cdots & b_{n-2} & \tr{\vec{b}}\vec{x} & \tr{\vec{b}}\vec{y}
\end{mx}\]

where $\vec{a}, \vec{b}, \vec{x}, \vec{y} \in \F^{n-2}$ are the vectors with corresponding coordinates $a_i, b_i, x_i, y_i$ for $1 \leq i \leq n-2$. In order to understand the spaces $S$ with basis $\vec{v}_1, \ldots, \vec{v}_{n-2}$, we seek conditions on $\vec{a}, \vec{b}$ which are consistent with matrices of this form having a vanishing permanent. The vectors $\vec{x}, \vec{y}$ represent the coefficients of possible linear combinations of basis vectors and are therefore allowed to vary freely. That is, we treat the permanent of the matrix above as a polynomial in $\vec{x}, \vec{y}$ with coefficients $\vec{a}, \vec{b}$. As long as this polynomial is not identically $0$, it is possible to find some values of $\vec{x}, \vec{y}$ which result in a nonzero permanent, implying that $S$ is not \vgood.

Using the definition of the permanent as a sum over permutations, we obtain 
\begin{align*}
\per\begin{mx}\tr{\vec{a}}\vec{x} & \tr{\vec{a}}\vec{y}\\ \tr{\vec{b}}\vec{x} & \tr{\vec{b}}\vec{y}\end{mx} &+ 
\sum_{1\leq i \leq n-2}x_i\per\begin{mx}a_i & \tr{\vec{a}}\vec{y} \\ b_i & \tr{\vec{b}}\vec{y}\end{mx} + y_i\per\begin{mx}a_i & \tr{\vec{a}}\vec{x} \\ b_i & \tr{\vec{b}}\vec{x}\end{mx} +\\
&+\sum_{\substack{1\leq i, j \leq n-2\\i\neq j}}x_iy_j\per\begin{mx}a_i & a_j \\ b_i & b_j\end{mx}
\end{align*}
and after collecting like terms we obtain
\begin{equation*}
\label{eq:cod2}
\tag{$\star$}
 6\sum_{1\leq i\leq n-2}x_iy_ia_ib_i+4\sum_{\substack{1\leq i, j \leq n-2\\i\neq j}}x_iy_j\per\begin{mx}a_i & a_j \\ b_i & b_j\end{mx}.
\end{equation*}
For this to be the $0$ polynomial, we must have $\Per_{(n-2)\times 2}\begin{mx}\vec{a} & \vec{b}\end{mx} = 0$ (recall that this notation means that every $2\times 2$ subpermanent vanishes). If the characteristic is $p = 3$, the first term of \eqref{eq:cod2} vanishes and this is the only condition, so the values $\vec{a}, \vec{b}$ are completely determined by the conditions in Lemma \ref{lem:i2}. 

On the other hand, if the characteristic is $p>3$, for the first term of \eqref{eq:cod2} to vanish we must also have $0 \in \{a_i, b_i\}$ for each $1 \leq i \leq n-2$. This means that conditions \ref{lemi2:c2} and \ref{lemi2:c3} of Lemma \ref{lem:i2} cannot be met, so that condition \ref{lemi2:c1} must be satisfied, i.e. $\vec{a} = 0$ or $\vec{b} = 0$. In the former case, $S \leq \{\vec{e}_3\}^\perp$ and in the latter $S\leq\{\vec{e}_4\}^\perp$. This reproves the conclusion of Theorem \ref{thm:trivialpermanull} for codimension-$2$ subspaces.

\begin{eg}[Example \ref{eg:nontrivial} revisited]
	Note that the basis vectors $\vec{v}_1, \vec{v}_2$ for $S \leq \F_3^4$,
	\begin{align*}
        &\vec{v}_1 = \tr{\begin{mx}1 & 0 & 1 & 1\end{mx}}
        &&\vec{v}_2 = \tr{\begin{mx}0 & 1 & 1 & -1\end{mx}}
    \end{align*}
    are in ``standard form'' with the first two coordinates $\vec{e}_1^2$ and $\vec{e}_2^2$ respectively. The last two coordinates satisfy the conditions of Lemma \ref{lem:i2}, as the third row of $\begin{mx}\vec{v}_1 & \vec{v}_2\end{mx}$ is of the form $(1, x)$ for $0 \neq x \in \F$ and the fourth row is of the form $(1, -x)$.
\end{eg}

One advantage of this proof over the one in Section \ref{sec:0fld} is that it also shows us what ``goes wrong'' if the characteristic is $3$. This will lead us to a full characterization of \vgood subspaces in any characteristic via the vanishing of the \emph{permanental polynomial}, a generalization of the polynomial \eqref{eq:cod2}.

For ease of computation, it will be convenient to introduce some notation; we continue to write vectors as column vectors. Fix some positive integers $d<k$. Because we shall need subscripts to distinguish between vectors, we write the components of a subscripted vector $\vec{a}_i \in \F^k$ as $\vec{a}_i = (a_{i, 1}, a_{i, 2}, \ldots, a_{i, k})^T$. As with the proof of the codimension-$2$ case, let $\vec{x}_1, \ldots, \vec{x}_d \in \F^{k}$ be some coefficient vectors and $\vec{a}_1, \ldots, \vec{a}_d\in \F^{k}$ variable vectors. The \emph{permanental polynomial} corresponding to this basis is the permanent of the $(k+d)\times(k+d)$ matrix $M$ below:

\[P := \per M \eqdef \per \begin{mx}
    I_{k} & \vec{x}_1 & \vec{x}_2 & \cdots & \vec{x}_d\\
    \vec{a}_1^T & \vec{a}_1^T\vec{x}_1 & \vec{a}_1^T\vec{x}_2 & \cdots & \vec{a}_1^T\vec{x}_d\\
    \vec{a}_2^T & \vec{a}_2^T\vec{x}_1 & \vec{a}_2^T\vec{x}_2 & \cdots & \vec{a}_2^T\vec{x}_d\\
    \vdots & \vdots & \vdots & \ddots & \vdots\\
    \vec{a}_d^T & \vec{a}_d^T\vec{x}_1 & \vec{a}_d^T\vec{x}_2 & \cdots & \vec{a}_d^T\vec{x}_d
\end{mx}\]
viewed as a polynomial in the variable vectors $\vec{x}_1, \ldots, \vec{x}_d$ (i.e.~in the $dk$ variables $x_{1,1}, \ldots, x_{d,k}$); note that it is a multilinear homogeneous polynomial of degree $d$, with each monomial containing at most one variable $x_{i, j}$ for each $1 \leq i \leq d$, and is also ``block-symmetric'' with respect to the variables $\vec{x}_1, \ldots, \vec{x}_d$. We next introduce notation for the monomials of $P$. 

We shall use Greek letters for ordered tuples. Given some tuple $\alpha \in [k]^{d}$ with components $\alpha = (\alpha_1, \ldots, \alpha_d)$, we denote by $\vec{x}^{\alpha}$ the monomial $x_{1, \alpha_1}x_{2,\alpha_2}\cdots x_{d,\alpha_d}$. 

\begin{lem}\label{lem:palpha}
    The coefficient of $\vec{x}^{\alpha}$ in $P$ is the same as the coefficient of $\vec{x}^{\alpha}$ in 
    \[P^{\alpha} \eqdef \per M^{\alpha} \eqdef \per \begin{mx}
        I_{k} & \vec{x}_1 & \vec{x}_2 & \cdots & \vec{x}_d\\
        \vec{a}_1^T & a_{1,\alpha_1}x_{1,\alpha_1} & a_{1,\alpha_2}x_{2,\alpha_2} & \cdots & a_{1,\alpha_d}x_{d,\alpha_d}\\
        \vec{a}_2^T & a_{2,\alpha_1}x_{1,\alpha_1} & a_{2,\alpha_2}x_{2,\alpha_2} & \cdots & a_{2,\alpha_d}x_{d,\alpha_d}\\
        \vdots & \vdots & \vdots & \ddots & \vdots\\
        \vec{a}_d^T & a_{d,\alpha_1}x_{1,\alpha_1} & a_{d,\alpha_2}x_{2,\alpha_2} & \cdots & a_{d,\alpha_d}x_{d,\alpha_d}
    \end{mx}\]
\end{lem}

\begin{proof}
    This is clear by examining the generalized diagonals: fix some $\sigma \in S_{k+d}$; the only difference between $M$ and $M^{\alpha}$ is the bottom right $d\times d$ submatrix, so suffice to examine those terms. Assume therefore that $M_{i,\sigma_i} = \vec{a}_{i}^T\vec{x}_{\sigma(i)} = \sum_{j=1}^{d}a_{i, j}x_{\sigma(i), j}$, after expanding the sums the only summand which contributes to $\vec{x}^{\alpha}$ is when $j = \alpha_{\sigma(i)}$, i.e. $a_{i, \alpha_{\sigma(i)}}x_{\sigma(i), \alpha_{\sigma(i)}} = M^{\alpha}_{i,\sigma_i}$.
\end{proof}

To determine $P$, we need to compute the coefficient of $\vec{x}^{\alpha}$ for every $\alpha \in [k]^d$. This coefficient turns out to be a multiple of 
\[A^{\alpha}\eqdef \per\begin{mx}
\vec{a}_1^{\alpha} & \vec{a}_2^{\alpha} & \cdots & \vec{a}_d^{\alpha}
\end{mx}.\]
The exact multiple depends on the choice of the tuple $\alpha \in [k]^d$ in the following manner: viewing $\alpha$ as a function $\alpha: [d] \to [k]$ assigning $i \in d$ to $\alpha_i \in [k]$, we obtain a partition of $[d]$ with blocks $\alpha^{-1}(j)$ (for $j \in \mathrm{img}(\alpha)$). Let $n_\alpha$ denote the number of ways of choosing at most one element from each block; if we denote the size of each block by $n_{j}\eqdef \crd{\alpha^{-1}(j)}$, then
\[n_{\alpha} = \prod_{j \in \mathrm{img}(\alpha)}(n_j+1).\]

For proving this claim, we introduce some helpful concepts. By a \emph{partial configuration from} $n_{\alpha}$ we mean a selection of at most one representative from each block $\alpha^{-1}(j)$ (for each $j \in \mathrm{img}(\alpha)$). (Thus there are precisely $n_{\alpha}$ many partial configurations.) The \emph{weight} of the partial configuration is the total number of representatives chosen. We say that a permutation $\tau \in S_{k+d}$ \emph{corresponds} to a partial configuration exactly when: the element $i \in \alpha^{-1}(j)$ was selected as part of the partial configuration if and only if $\tau(\alpha_i) = k+i$. Note that every permutation corresponds to some partial configuration (possibly the empty one).

\begin{thm}\label{thm:ppol}
    The coefficient of $\vec{x}^{\alpha}$ in $P$ is $n_\alpha A^{\alpha}$.
\end{thm}

\begin{proof}
    Fix some $\alpha \in [k]^d$. By Lemma \ref{lem:palpha}, suffice it to determine the coefficient of $\vec{x}^{\alpha}$ in $P^{\alpha}$. By the definition of the permanent, $P^{\alpha} = \sum_{\sigma \in S_{k+d}}\prod_{i=1}^{k+d}M^{\alpha}_{i,\sigma(i)}$. For the variable $x_{i, \alpha_i}$ to appear in the term $\prod_{i=1}^{k+d}M^{\alpha}_{i,\sigma(i)}$ exactly one of the following two possibilities must be satisfied:
    \begin{itemize}
        \item either $k+i = \sigma(\alpha_i)$; or
        \item there exists some $k+1\leq j\leq k+d$ such that $k+i = \sigma(j)$.
    \end{itemize}
    Call a permutation \emph{valid} if it satisfies exactly one of these conditions for each $1 \leq i \leq d$, and let $S_{k+d}^{\alpha}$ be the set of valid permutations.

    Fix some partial configuration of $n_\alpha$. We shall prove that the sum over all valid permutations $\sigma \in S_{k+d}^{\alpha}$ corresponding to this partial configuration is precisely 
    \[\vec{x}^{\alpha}A^{\alpha} = \vec{x}^{\alpha}\per\begin{mx}
    \vec{a}_1^{\alpha} & \vec{a}_2^{\alpha} & \cdots & \vec{a}_d^{\alpha}
    \end{mx}.\]
    Since every permutation corresponds to some partial configuration and there are $n_{\alpha}$ partial configurations, this immediately implies that the coefficient of $\vec{x}^{\alpha}$ in $P$ is precisely $n_{\alpha}A^{\alpha}$, as we wanted to show.

    To prove the claim, suppose the partial configuration is of weight $\ell$ and assume without loss of generality that the partial configuration is such that $1, \ldots, \ell$ are the chosen representatives (the permanent is invariant under permutation of columns) so any corresponding permutation must have $\sigma(\alpha_1) = k+1$, $\sigma(\alpha_2) = k+2$, $\ldots$, $\sigma(\alpha_\ell) = k+\ell$. This gives us $x_{1, \alpha_{1}}\cdots x_{\ell, \alpha_\ell}$, and we must determine the remainder of $\sigma$.

    We claim that for $1 \leq r \leq k$ such that $r \notin \{\alpha_1, \ldots, \alpha_\ell\}$, we must have $\sigma(r) = r$. To prove this, suffice it to show that $1\leq \sigma(r)\leq k$, for then $M^{\alpha}_{r, \sigma(r)}$ is an entry of $I_k$ and must be $1$ (i.e., $\sigma(r) = r$) lest the whole term be zero. However, if $\sigma(r)\geq k+1$, say $k+t = \sigma(r)$, then (since $\sigma$ is a valid permutation and $1 \leq r \leq k$) we must have $r = \alpha_t$ and this is impossible since $r \notin \{\alpha_1, \ldots, \alpha_\ell\}$ and $\sigma$ corresponds to the partial configuration.

    We have now accounted for the first $k$ rows of $M^{\alpha}$ and also for some $k$ columns: the columns $k+1, \ldots, k+\ell$ and the columns $\{1, \ldots, k\}\setminus\{\alpha_1, \ldots, \alpha_\ell\}$. The remainder of $\sigma$ can be chosen freely (under the constraint that the result is a permutation); thus $\prod_{j=k+1}^{k+d}M^{\alpha}_{j, \sigma(j)}$ is a generalized diagonal of the matrix

    \[B = \begin{mx}
        a_{1, \alpha_1} & \cdots & a_{1, \alpha_\ell} & a_{1, \alpha_{\ell+1}}x_{\ell+1, \alpha_{\ell+1}} & \cdots & a_{1, \alpha_{d}}x_{d, \alpha_d}\\
        a_{2, \alpha_1} & \cdots & a_{2, \alpha_\ell} & a_{2, \alpha_{\ell+1}}x_{\ell+1, \alpha_{\ell+1}} & \cdots & a_{2, \alpha_d}x_{d, \alpha_d}\\
        \vdots & \ddots & \vdots & \vdots & \ddots & \vdots \\
        a_{d, \alpha_1} & \cdots & a_{d, \alpha_\ell} & a_{d, \alpha_{\ell+1}}x_{\ell+1, \alpha_{\ell+1}} & \cdots & a_{d, \alpha_d}x_{d, \alpha_d}
    \end{mx}\]
    and summing over the possibilities for $\sigma$ we get $x_{\ell+1, \alpha_{\ell+1}}\cdots x_{d, \alpha_d}A^{\alpha}$.
\end{proof}

\begin{crl}
    For any injective $\alpha: [d] \to [k]$, the coefficient of $\vec{x}^\alpha$ in 
    \[P = \begin{mx}
    I_{k} & \vec{x}_1 & \vec{x}_2 & \cdots & \vec{x}_d\\
    \vec{a}_1^T & \vec{a}_1^T\vec{x}_1 & \vec{a}_1^T\vec{x}_2 & \cdots & \vec{a}_1^T\vec{x}_d\\
    \vec{a}_2^T & \vec{a}_2^T\vec{x}_1 & \vec{a}_2^T\vec{x}_2 & \cdots & \vec{a}_2^T\vec{x}_d\\
    \vdots & \vdots & \vdots & \ddots & \vdots\\
    \vec{a}_d^T & \vec{a}_d^T\vec{x}_1 & \vec{a}_d^T\vec{x}_2 & \cdots & \vec{a}_d^T\vec{x}_d
\end{mx}\]
    is $2^d A^{\alpha}$.
\end{crl}

\begin{eg}
    To illustrate Theorem \ref{thm:ppol}, let us compute the permanental polynomial 
    \[P_3 \eqdef \per \begin{mx}
    & & & & & x_1 & y_1 & z_1 \\
    & & & &  & x_2 & y_2 & z_2\\
    & & & I_{n-3} & &  x_3 & y_3 & z_3 \\
    & & & & & \vdots & \vdots & \vdots \\
    & & & & & x_{n-3} & y_{n-3} & z_{n-3} \\
    a_1 & a_2 & a_3 & \cdots & a_{n-3} & \tr{\vec{a}}\vec{x} & \tr{\vec{a}}\vec{y} & \tr{\vec{a}}\vec{z}\\
    b_1 & b_2 & b_3 & \cdots & b_{n-3} & \tr{\vec{b}}\vec{x} & \tr{\vec{b}}\vec{y} & \tr{\vec{b}}\vec{z}\\
    c_1 & c_2 & c_3 & \cdots & c_{n-3} & \tr{\vec{c}}\vec{x} & \tr{\vec{c}}\vec{y} & \tr{\vec{c}}\vec{z}
\end{mx}\]
    The monomials are of the form $x_{r}y_{s}z_{t}$ and there are three types of terms according to the number of distinct indices $r, s, t$:
    \begin{align*}
        P_3 =\, &4\sum_{1\leq r\leq n-3}x_ry_rz_r\per\begin{mx}
            a_r & a_r & a_r\\
            b_r & b_r & b_r\\
            c_r & c_r & c_r
        \end{mx}\\
        &+ 6\sum_{\substack{1\leq r, s\leq n-3\\r\neq s}}(x_ry_rz_s+x_ry_sz_r+x_sy_rz_r)\per\begin{mx}
            a_r & a_r & a_s \\
            b_r & b_r & b_s \\
            c_r & c_r & c_s
        \end{mx}\\
        &+ 8\sum_{\substack{1\leq r, s, t\leq n-3\\r, s, t \mbox{ \tiny distinct}}}x_ry_sz_t\per\begin{mx}
            a_r & a_s & a_t\\
            b_r & b_s & b_t\\
            c_r & c_s & c_t
        \end{mx}.
    \end{align*}

    Suppose now that we are working over a field of characteristic $p>3$ (and therefore of characteristic $p>3+1$), so that the coefficient of each monomial above is $0$ if and only if the permanent of the corresponding $3\times 3$ matrix is $0$. One can show, via case analysis, that one of $\vec{a}, \vec{b}, \vec{c}$ must be $0$. This would reprove the conclusion of Theorem \ref{eg:nontrivial} for codimension-$3$ subspaces; i.e.~that if $S\leq \F_q^n$ is a \vgood space of codimension $3$, then it is trivial. We shall not go through the case analysis, but instead derive this from a more general theorem below.
\end{eg}

\begin{thm}\label{thm:zeropoly}
	Let $S \leq \F^n$ be a subspace of codimension $d$. Permuting the coordinates of $\F^n$ if necessary, fix a \emph{standard basis} for $S$ of the form 
	\[(\vec{e}_1, \vec{a}_1), \ldots, (\vec{e}_{n-d}, \vec{a}_{n-d})\]
	(where $\vec{e}_i \in \F^{n-d}$ and $\vec{a}_i\in\F^{d}$). Let $P$ be the corresponding $n\times n$ permanental polynomial
	\[P \eqdef \per \begin{mx}
    &  & & x_{1,1} & \cdots & x_{d,1} \\
    &  I_{n-d} & & \vdots & \ddots & \vdots \\
    &  & & x_{1,n-d} & \cdots & x_{d,n-d} \\
    a_{1,1} & \cdots & a_{1,n-d} & \tr{\vec{a}_1}\vec{x}_1 & \cdots & \tr{\vec{a}_1}\vec{x}_d\\
    \vdots & \ddots & \vdots & \vdots & \ddots & \vdots\\
    a_{d,1} & \cdots & a_{d,n-d} & \tr{\vec{a}_d}\vec{x}_1 & \cdots & \tr{\vec{a}_d}\vec{x}_d
\end{mx}.\]
	Then $S$ is \vgood if and only if $P$ is the zero polynomial.
\end{thm}

\begin{proof}
	Let us start with the forward direction and suppose $P$ is identically zero. Given any $n$ vectors $\vec{u}_1, \ldots, \vec{u_n} \in S$ let us express them as a linear combination of the basis elements $\vec{v}_1, \ldots, \vec{v}_{n-d}$
	\[\vec{u}_i = \sum_{j=1}^{n-d}\lambda_{i,j}\vec{v}_j.\]
	We wish to prove $\per\begin{mx}\vec{u}_1 &\cdots &\vec{u}_n\end{mx} = 0$. By the multilinearity of the permanent, suffice it to show that
	\begin{align*}
		\per\begin{mx}\vec{v}_{j_1} & \cdots & \vec{v}_{j_n}\end{mx} = 0
	\end{align*}
	for every choice of $j_1, \ldots, j_n \in \{1, \ldots, n-d\}$ (where distinct $j_k$ may correspond to the same index).

	Note that if one of the indices $1, \ldots, n-d$ is not selected, then the resulting matrix has a row of $0$ and so vanishing permanent. Suppose therefore that each of $1, \ldots, n-d$ is selected. Therefore, after permuting the columns, we obtain a permanent of the form 
	\[P \eqdef \per \begin{mx}
    &  & & x_{1,1} & \cdots & x_{d,1} \\
    &  I_{n-d} & & \vdots & \ddots & \vdots \\
    &  & & x_{1,n-d} & \cdots & x_{d,n-d} \\
    a_{1,1} & \cdots & a_{1,n-d} & \tr{\vec{a}_1}\vec{x}_1 & \cdots & \tr{\vec{a}_1}\vec{x}_d\\
    \vdots & \ddots & \vdots & \vdots & \ddots & \vdots\\
    a_{d,1} & \cdots & a_{d,n-d} & \tr{\vec{a}_d}\vec{x}_1 & \cdots & \tr{\vec{a}_d}\vec{x}_d
	\end{mx}\]
	with $\vec{x}_1, \ldots, \vec{x}_d \in \{\vec{e}_1, \ldots, \vec{e}_{n-d}\}$. Since $P$ is the zero polynomial, we are done.

	For the backward direction, suppose $S$ is \vgood subspace. Then the polynomial $P$ vanishes for each choice of the $d(n-d)$ variables $x_{1, 1}, \ldots, x_{d, n-d}$. But $P$ is multilinear and homogeneous of degree $d$ and so must be the zero polynomial.
\end{proof}

\begin{crl}\label{crl:computation}
    Let $S \leq \F^n$ be a subspace of codimension $d$. It is possible to determine whether $S$ is \vgood using at most $O_d(n^d)$ operations.
\end{crl}

\begin{proof}
    Permuting the coordinates of $\F^n$ if necessary, fix a \emph{standard basis} for $S$ of the form 
    \[(\vec{e}_1, \vec{a}_1), \ldots, (\vec{e}_{n-d}, \vec{a}_{n-d})\]
    (where $\vec{e}_i \in \F^{n-d}$ and $\vec{a}_i\in\F^{d}$), and let $P$ be the corresponding $n\times n$ permanental polynomial. By Theorem \ref{thm:zeropoly}, $S$ is \vgood if and only if $P$ is the zero polynomial. By Theorem \ref{thm:ppol} we can verify this by computing $n^d$ permanents of size $d\times d$.
\end{proof}

\begin{crl}\label{crl:sufficient}
	Let $S \leq \F^n$ be a subspace of codimension $d$. Permuting the coordinates of $\F^n$ if necessary, fix a \emph{standard basis} for $S$ of the form 
	\[\vec{v}_1 = (\vec{e}_1, \vec{a}_1), \ldots, \vec{v}_{n-d} = (\vec{e}_{n-d}, \vec{a}_{n-d})\]
	(where $\vec{e}_i \in \F^{n-d}$ and $\vec{a}_i\in\F^{d}$). Let $\widehat{S}:= \spn\{\vec{a}_1, \ldots, \vec{a}_{n-d}\} \leq \F^{d}$. If $\widehat{S}$ is \vgood then $S$ is \vgood.
\end{crl}

\begin{proof}
	Suppose $\widehat{S}$ is \vgood, so that every $d\times d$ permanent with columns from $\widehat{S}$ vanishes. By Theorem \ref{thm:ppol} every monomial in the permanental polynomial corresponding to the basis $\vec{v}_1, \ldots, \vec{v}_{n-d}$ vanishes, so that by Theorem \ref{thm:zeropoly} $S$ is \vgood.
\end{proof}

\begin{rk}
Example \ref{eg:nontrivial} shows that the sufficient condition in Corollary \ref{crl:sufficient} is not necessary, as $\widehat{S}$ in that example is all of $\F^d$, which cannot be \vgood. 
\end{rk}

Theorem \ref{thm:zeropoly} ``explains'' the characteristic bound $\chr\F > k+1$ of Theorem \ref{thm:trivialpermanull}. We conclude the paper with another proof of this theorem.

\begin{proof}[Second proof of Theorem \ref{thm:trivialpermanull}]
	The cases $n=1, 2$ are straightforward. Suppose the theorem holds for all dimensions up to (but not including) some fixed $n$, and we shall prove it for $n$. Towards that end, let $S\leq \F^{n}$ be a \vgood subspace of codimension $d$ and suppose that $\chr\F>d+1$.

	Permuting the coordinates of $\F^{n}$ if necessary, fix a standard basis for $S$ of the form 
	\[\vec{v}_1 = (\vec{e}_1, \vec{a}_1), \ldots, \vec{v}_{n-d} = (\vec{e}_{n-d}, \vec{a}_{n-d})\]
	(where $\vec{e}_i \in \F^{n-d}$ and $\vec{a}_i\in\F^{d}$). Let $\widehat{S}:= \spn\{\vec{a}_1, \ldots, \vec{a}_{n-d}\} \leq \F^{d}$. Since we may assume $d < n$ (the case $d=n$ being trivial) suffice it to show that $\widehat{S}$ is a \vgood subspace of $\F^d$, since then inductive hypothesis implies that $\vec{a}_1, \ldots, \vec{a}_{n-d}$ share a common zero coordinate, which immediately implies that $\vec{v}_1, \ldots, \vec{v}_{n-d}$ share a common zero coordinate.

	In order to show that $\widehat{S}$ is \vgood, let $\vec{b}_1, \ldots, \vec{b}_d$ be arbitrary vectors of $\widehat{S}$ and write each as a linear combination of the spanning vectors
	\[\vec{b}_i = \sum_{j=1}^{n-d}\lambda_{i, j}\vec{a}_j.\]
	As with the proof of Theorem \ref{thm:trivialpermanull}, in order to show $\per\begin{mx}\vec{b}_1 & \cdots & \vec{b}_d\end{mx} = 0$, by the multilinearity of the permanent, suffice it to show that each permanent of the following form vanishes
	\begin{align*}
		\per\begin{mx}\vec{a}_{j_1} & \cdots & \vec{a}_{j_d}\end{mx} = 0
	\end{align*}
	for every choice of $j_1, \ldots, j_d \in \{1, \ldots, n-d\}$ (where distinct $j_k$ may correspond to the same index). Since $\chr\F > d+1$ Theorem \ref{thm:ppol} implies that the coefficient of each monomial in the permanental polynomial corresponding to the basis $\vec{v}_1, \ldots, \vec{v}_{n-d}$ is nonzero. Since $S$ is assumed to be \vgood, Theorem \ref{thm:zeropoly} implies that this permanental polynomial is the zero polynomial, so it must be that each monomial is zero. Again, by Theorem \ref{thm:ppol} this means that 
	\begin{align*}
		\per\begin{mx}\vec{a}_{j_1} & \cdots & \vec{a}_{j_d}\end{mx} = 0
	\end{align*}
	for every choice of $j_1, \ldots, j_d \in \{1, \ldots, n-d\}$.
\end{proof}

\newpage
\bibliographystyle{abbrv}
\bibliography{particle.bib}
\end{document}